\documentclass[12pt,english]{article}
\usepackage{geometry}
\usepackage{color}
\usepackage{babel}
\usepackage{amsmath}
\usepackage{amsthm}
\usepackage{amsfonts}
\usepackage{amssymb}
\usepackage{graphicx}
\usepackage[longnamesfirst,sort,comma]{natbib}
\usepackage[unicode=true,pdfusetitle,
 bookmarks=true,bookmarksnumbered=false,bookmarksopen=false,
 breaklinks=false,pdfborder={0 0 1},colorlinks=true]{hyperref}
\hypersetup{linkcolor=red, citecolor=blue}
\usepackage{enumitem}
\usepackage{bm}
\theoremstyle{plain}
\newtheorem{thm}{\protect\theoremname}[section]
\theoremstyle{plain}
\newtheorem{prop}[thm]{\protect\propositionname}
\ifx\proof\undefined
\newenvironment{proof}[1][\protect\proofname]{\par
	\normalfont\topsep6\p@\@plus6\p@\relax
	\trivlist
	\itemindent\parindent
	\item[\hskip\labelsep\scshape #1]\ignorespaces
}{%
	\endtrivlist\@endpefalse
}
\providecommand{\proofname}{Proof}
\fi
\theoremstyle{remark}

\theoremstyle{definition}
\newtheorem{defn}[thm]{\protect\definitionname}
\theoremstyle{plain}
\newtheorem{lem}[thm]{\protect\lemmaname}
\newtheorem{hp}{Assumption}

\providecommand{\definitionname}{Definition}
\providecommand{\lemmaname}{Lemma}
\providecommand{\propositionname}{Proposition}
\providecommand{\remarkname}{Remark}
\providecommand{\theoremname}{Theorem}

\newcommand{\cvd}{\mbox{$\stackrel{d}{\longrightarrow}\,$}}

\newcommand{\cvp}{\mbox{$\stackrel{p}{\longrightarrow}\,$}}

\newcommand{\E}{\operatorname{E}}
\newcommand{\1}{\mathbf{1}}
\def\Real{\mathbb{R}}
\newcommand{\as}{\mbox{ a.s.}}

\let\what\widehat
\setcitestyle{citesep={;}}

\renewcommand{\baselinestretch}{1.2}

\begin{document}
\title{Encompassing Tests for Nonparametric Regressions\thanks{%We would like to thank the Editor Michael Jansson and two anonymous referees for their comments that have significantly improved the quality of the paper. We are also grateful to Jean-Pierre Florens and Christian Gouriéroux for their feedback.
P. Lavergne acknowledges funding from ANR under grant
	ANR-17-EURE-0010 (Investissements d'Avenir program).
}}
\author{Elia Lapenta\thanks{\emph{CREST and ENSAE. Email:
	elia.lapenta@ensae.fr} Address correspondence: CREST, 5 Avenue
	Le Chatelier, 91120 Palaiseau, FRANCE.} \hspace{.1cm}  and
\ Pascal Lavergne\thanks{\emph{Toulouse School of Economics.
Email: lavergnetse@gmail.com} Address correspondence: Toulouse School
of Economics, 1 Esplanade de l'Universit\'e, 31080 Toulouse Cedex 06,
FRANCE.}}
\date{July 18, 2023}

\maketitle
\thispagestyle{empty}
\vspace*{-.5cm}
\normalsize

\begin{abstract}
We set up a formal framework to characterize encompassing of
nonparametric models through the $L^2$ distance.  We contrast it to
previous literature on the comparison of nonparametric regression
models. We then develop testing procedures for the encompassing
hypothesis that are fully nonparametric.  Our test statistics depend
on kernel regression, raising the issue of bandwidth's choice.  We
investigate two alternative approaches to obtain a “small bias
property” for our test statistics. We show the validity of a wild
bootstrap method.  We empirically study the use of a data-driven
bandwidth and illustrate the attractive features of our tests for
small and moderate samples.
\end{abstract}

Locally Robust Statistic.
\vspace{0.25cm}

\newpage
\setcounter{page}{2}

\section{Introduction}

The encompassing principle was introduced in econometrics by
\cite{hendry_formulation_1982},
\citet{gourieroux_testing_1983}, and \cite{mizon_encompassing_1986},
and further developed in \citet{gourieroux_testing_1995},
 \citet{florens_encompassing_1996}, and
 \cite{dhaene_instrumental_1998} among others.  It provides a natural
 principle for choosing between two competing theories:  a new theory
 must be able to accommodate the results obtained by a concurrent
 older one. An extensive survey is provided in
 \citet{bontemps_encompassing:_2008}.

Our goal is to propose encompassing tests for nonparametric models.
Our main steps are (i) to formally define encompassing
 for nonparametric models, (ii) to develop fully nonparametric
 encompassing tests, and (iii) to show asymptotic validity of a wild
 bootstrap method for asymptotic inference.  Our first contribution is
 thus to formally set up a framework to precisely define encompassing
 for nonparametric regression models. We discuss nonparametric
 encompassing with respect to previous literature on the comparison of
 such models, whether nested or non nested, see below for references.
 We show that encompassing reduces neither to significance of some
 variables nor to the comparison of the models' theoretical
 fit. Hence, the null hypothesis of encompassing cannot generally be
 tested with existing procedures.  Our second contribution is to
 propose fully nonparametric encompassing tests for regression models.
 Existing encompassing tests rely on parametric functional forms,
 except for
\cite{bontemps_parametric_2008} who propose a test aimed at assessing
a consequence of encompassing.  The new tests we develop directly test
the encompassing hypothesis.  They are based on an empirical process
estimating a continuum of unconditional moments following the
Integrated Conditional Moment (ICM) principle introduced by
\citet{bierens_consistent_1982}.  Our third contribution is to develop
a wild bootstrap method and to show that it provides asymptotically
correct inference.

Our fourth contribution is to propose and investigate two approaches
to obtain a “small bias property.” Our test statistic depends on an
empirical process involving a first-step nonparametric estimator.  The
small bias property of a semiparametric estimator is that its bias
converges to zero faster than the pointwise and integrated bias of the
nonparametric estimator on which it is based. A distinguishing feature
is that the resulting statistic is $\sqrt{n}$-consistent even when the
nonparametric estimator on which it is based converges at the optimal
nonparametric rate. Without this property, using a first-step
nonparametric estimator necessitates some undersmoothing, and this
complicates practical implementation.
\cite{newey_twicing_2004} have developed a generic technique based on
twicing kernels to obtain semiparametric estimators with small bias.
 We develop here two alternative methods that yield a small bias
 property. The first one uses a
 bias-corrected kernel estimator based on the boosting principle
 \citep{di_marzio_boosting_2008, park_l2_2009}. The second approach is
 to make the empirical process of interest {\em locally robust} with
 respect to the nonparametric regression.  This has previously been
 used successfully in semiparametric estimation, see
 \cite{newey_semiparametric_1990} for an early example, and
 \cite{Chernozhukov2022} for a general approach.  We here
 adapt the two approaches to our empirical process of interest and we show that
 these  yield a small bias property in the asymptotic
 expansion of our test statistic.  This allows for a larger set of
 smoothing parameters, so the test is expected to be less sensitive to
 the bandwidth choice.

Our work is related to the extensive literature on consistent
specification testing based on empirical processes, see
\cite{bierens_asymptotic_1997}, \citet{stinchcombe_consistent_1998},
\citet{xia_goodness--fit_2004}, \cite{escanciano_consistent_2006},
\cite{delgado_distribution-free_2008},
\cite{lavergne_breaking_2008} to mention just a few.
The main features of our tests compared to previous work are that (i)
our empirical process contains a nonparametric kernel estimator, and
(ii) due to the form of the null hypothesis we cannot use a
density-weighted process, and we thus need to control for a random
denominator.  Our nonparametric encompassing tests are also connected
to the comparison of nonparametric regressions in nested and non
nested cases,  e.g., \citet{fan_consistent_1996}, \citet{LV96},
\citet{delgado_significance_2001}, and \citet{lavergne_significance_2015}.
We show however that the encompassing hypothesis cannot be tested
through existing procedures. Our work is also related to the literature on
estimation and testing with nonparametric nuisance components, e.g.,
\cite{escanciano_uniform_2014} and \cite{mammen_semiparametric_2016}.
The former authors obtain uniform-in-bandwidth expansions for an
empirical process similar to the one we consider.  We focus here on
obtaining a small bias property, but we do not formally establish
uniformity in bandwidth.

Our paper is organized as follows. Section
\ref{sec:Principle} formalizes the encompassing
notion for nonparametric models and compares our framework to the
literature on comparison of nonparametric regressions.  Section
\ref{sec:Tests} details the
construction of the test statistics and the two approaches used to
obtain a small bias property.  Section
\ref{sec:Theory} is devoted to the analysis of
the asymptotic behavior of our statistics. Since the asymptotic
distribution under the null depends on unknown features of the DGP, we
establish in Section \ref{sec:Bootstrap} the validity of a wild
bootstrap procedure. Section \ref{sec:Simulations} provides evidence
about the small sample performances of our procedures. We check that
bootstrapping allows to correctly control size and that our tests have
good power.  We evaluate the benefits of our bias-reducing approaches, and
we investigate thoroughly the influence of the bandwidth as well as of
the trimming parameter, which is theoretically necessary. We also
provide an empirical illustration. Section \ref{sec:Proofs} contains
the proofs of our main results.  A supplementary material contains the
proof of a technical lemma.

\section{Encompassing for Nonparametric Regressions}
\label{sec:Principle}

The definition of encompassing
starts with the definition of the {\em binding function}, see
\citet{gourieroux_testing_1995}.  In a parametric context, we
typically start with two competing parameterized families of densities
 for $Y$,
${\mathcal{M}_{1}}=\{g_{1}(\cdot,\alpha_{1})\, : \,  \alpha_{1}\in
 A_{1}\}$ and ${\mathcal{M}_{2}}:=\{g_{2}(\cdot,\alpha_{2})\, : \,
\alpha_{2}\in A_{2}\}$.
The pseudo true value of $\alpha_{i}$, $i=1$ or $2$, is defined as
\[
\alpha_{i}^{*}=
\arg\min_{\alpha_{i}\in A_{i}}
d (f(\cdot) , g_{i}(\cdot,\alpha_{i}))
\, ,
\]
where $f(\cdot)$ is the true  density of $Y$ and $d$ is
some divergence between the two distributions,  for
instance the Kullback-Leibler divergence
%defined as
$
%d_{KL} (f(\cdot) , g(\cdot,\alpha))=
\E_f \left[ \log \left( f(\cdot) / g(\cdot,\alpha)\right) \right]
$,
with expectation taken with respect to $f(\cdot)$.
The binding function $b(\alpha_{1})$ is a correspondence between an element of model
${\mathcal{M}_{1}}$ and the  element of model
${\mathcal{M}}_{2}$ that is closest to it.  Specifically,
\[
b(\alpha_{1}) =
\arg\min_{\alpha_{2}\in A_{2}}
d (g_1(\cdot,\alpha_{1}), g_{2}(\cdot,\alpha_{2}))
\, .
\]
We then say that  ${\mathcal{M}_{1}}$ encompasses model
${\mathcal{M}_{2}}$ if
$
\alpha_{2}^{*}=b(\alpha_{1}^{*})
$.
That is, ${\mathcal{M}_{1}}$ encompasses
${\mathcal{M}_{2}}$ if the pseudo-true value of the latter can
be obtained from the pseudo-true value of the former.

Here we focus on two nonparametric competing models to explain $Y$,
where Model $\mathcal{M}_{W}$ uses covariates $W$ and Model
$\mathcal{M}_{X}$ uses $X$. A nonparametric model with a specific set
of covariates is a function of these variables.  To define the
pseudo-true value that corresponds to the best explanation of $Y$, we
consider the ${L}^{2}$ distance. That is, $\mathcal{M}_{W}$
and $\mathcal{M}_{X}$ are respectively defined as
${L}^{2}(W)$, the space of the square integrable functions of
$W$, and ${L}^{2}(X)$.  The “pseudo-true functions” are then
the regression functions
\begin{align}
m(W) & = \arg\min_{g \in {L}^{2}(W)} \E \left( Y-g(W)
\right)^{2} = \E (Y|W),
\nonumber
\\
\mbox{and }
m_X(X) & = \arg\min_{h \in {L}^{2}(X)} \E \left( Y-h(X) \right)^{2}
= \E (Y|X)
\, .
\label{eq: Pseudo true value in Encompassing of this paper}
\end{align}
The binding function is similarly defined  as
\[
b( g(W) )= \arg\min_{h \in {L}^{2}(X)}
\E \left( g(W) - h(X) \right)^{2}
 =
 \E \left( g(W) | X\right)
\, .
\]
We  thus say that  ${\mathcal{M}_{W}}$ encompasses model ${\mathcal{M}_{X}}$
if
\begin{equation}
 b(m (W)) = m_X(X) \as \Leftrightarrow
\E \left( \E (Y|W) | X\right) =  \E (Y|X)
\as
%\, .
\label{eq:EncReg}
\end{equation}
That is, the regression function of $Y$ on $X$ can be obtained from the
regression function of $Y$ on $W$.
Similarly, ${\mathcal{M}_{X}}$ encompasses model
${\mathcal{M}_{W}}$ if $\E \left( \E (Y|X) | W\right) =  \E
(Y|W)$ almost surely.

In what follows, we explore the implications of the definition of
nonparametric encompassing and relate it to previous work on the
comparison of regression models.  In particular, we show that
encompassing reduces neither to significance of some variables nor to
the comparison of the models' theoretical fit. Hence, the null
hypothesis of encompassing cannot generally be tested with existing
procedures.

\subsection{Encompassing and Model Fit}

\cite{LV96} proposed comparing nonparametric regression models on the
basis of their theoretical fit and considered the
hypotheses
\begin{align*}
H_0\, :  \, & \E\left[Y - m(W)\right]^2-  \E\left[Y -m_X(X)\right]^2 =
0
\, ,
\\
H_W\, : \, & \E\left[Y - m(W)\right]^2-  \E\left[Y -m_X(X)\right]^2 <
0
\, ,
\\
H_X\, :  \,  & \E\left[Y - m(W)\right]^2-  \E\left[Y -m_X(X)\right]^2 >
0
\, .
\end{align*}
Non rejection of the null hypothesis $H_0$ means that both models have
the same theoretical fit. Rejection of $H_0$ in favor of either $H_W$
or $H_X$ indicates which model dominates the other.
This framework is quite general since it does not make a distinction
between nested and non nested situations and treats the two competing
models symmetrically. \cite{LV96} built a test of $H_{0}$ against
$H_W$ and $H_X$ based on the comparison of the empirical analogs of the models' fit.

 Comparing nonparametric regressions through their fit seems
 natural. Does encompassing imply a better fit for the
 encompassing model? As we detail below, the answer is yes: if
 $\mathcal{M}_{W}$ encompasses $\mathcal{M}_{X}$, the theoretical fit
 of $\E(Y|W)$ is at least as good as the one of $\E(Y|X)$, and is
 strictly better except when $\E(Y|W) = \E(Y|X) \as$
%This is because encompassing is defined through the regression functions.
\begin{prop}
\label{prop: Better explanation by encompassing}
If  $\mathcal{M}_{W}$ encompasses  $\mathcal{M}_{X}$,
\begin{enumerate}[label=(\alph*)]
\item $H_X$ cannot hold,
\item  $H_0$ holds iff  $\E(Y|W) = \E(Y|X) \as$ iff
$\mathcal{M}_{X}$ encompasses $\mathcal{M}_{W}$.
\end{enumerate}
\end{prop}
\begin{proof}
Consider  Statement (a).
Since $\E (  \E(Y|W) | X) =  \E(Y|X) \as$,
\begin{align*}
E\left[Y - \E(Y|X)\right]^2 &  =
\E\left[Y - \E(Y|W) \right]^2 +
\E\left[\E(Y|W)-  \E(Y|X) \right]^2
\, ,
\end{align*}
as the cross-term cancels, and
$E\left[Y - \E(Y|X)\right]^2  \geq
\E\left[Y - \E(Y|W) \right]^2$.

Consider  Statement (b). It is obvious that if $\E(Y|W) =
\E(Y|X) \as$ then $H_0$ holds and both models encompass each other.
By \citet[Lemma 1]{LV96}, if one model encompasses the other and
$H_0$ holds, it should be that $\E(Y|W) = \E(Y|X)$.
Finally, if  both models encompass each other, then $H_0$ holds by
 Statement (a).
\end{proof}

Therefore, except in the case where the two regressions are equal,
encompassing implies a strictly better fit for the encompassing model.
%If the two regressions are equal, and thus encompass each other, then
%each regression depends solely on the set of variables that are common
%to $W$ and $X$ (or more generally on the intersection of the
%sigma-algebra generated by the two sets).
Conversely, one model can have a better fit than another without
encompassing it, see \cite{LV96} for some examples where this
happens. Hence, the encompassing concept is not tailored for model
selection.
 
The test statistic proposed by \cite{LV96} to compare models' fit is
asymptotically degenerate under $H_0$ when the two models are
“generalized nested regressions,” that is when either $\E ( \E(Y|W) |
X) = \E(Y|X) \as$ or $\E ( \E(Y|X) | W) = \E(Y|W) \as$ Hence their
test cannot be used when there is encompassing, and the tests we
develop below are complementary to theirs. The latter might also be
used as preliminary tests to check whether their test can be
entertained. This would imply (i) testing whether Model ${\cal M}_X$
encompasses Model ${\cal M}_W$, and (ii) testing whether Model ${\cal
M}_W$ encompasses Model ${\cal M}_X$. One rejection among two would
imply (we cannot reject) that the encompassing model has a strictly
better fit than the encompassed one; no rejection would mean (we
cannot reject) that $\E(Y|W) = \E(Y|X) \as$ Only if we rejected twice
would we need to entertain the test of \cite{LV96} to determine
whether one model has a better fit.  A more direct alternative to this
involved procedure is to apply the test proposed
by \cite{liao2020nondegenerate}, which is universally valid.  We leave
the study of the respective merits of these competing procedures for
future work.

\subsection{Encompassing  and Significance}
\label{sec:EncSignif}
When the two sets of regressors are nested, specifically when $W= (X,
Z)$, then it is clear that (\ref{eq:EncReg}) holds. A more interesting
question is whether ${\mathcal{M}_{X}}$ encompasses
${\mathcal{M}_{W}}$, that is whether
\[
\E \left( \E (Y | X) | X, Z \right)  = \E(Y|X) =\E (Y| X, Z) \as
%\Leftrightarrow
%\E (Y | X, Z) =  \E (Y| X) \as
%\, .
\]
In this setup, encompassing is equivalent to whether $Z$ is significant in the
regression function of $Y$ on $X$ and $Z$. Significance testing in
nonparametric regressions  has been a focus
of extensive work, see \cite{fan_consistent_1996,
lavergne2000nonparametric, Abs01, Lav2001, delgado_significance_2001}. 

Consider now two nonnested sets of regressors,
and assume the regressors $X$ are not significant once we control for $W$,
that is
\begin{equation}
\E (Y|W,X)=\E (Y|W) \as
%\, .
\label{hyp:sign}
\end{equation}
By conditioning again with respect to $X$, one obtains that
\[
\E (Y|X)   = \E [ \E (Y|W) |X] \as
%\, ,
\]
so that $\mathcal{M}_{W}$ encompasses $\mathcal{M}_{X}$.
\cite{bontemps_parametric_2008} consider (\ref{hyp:sign}) as their
hypothesis of interest, since it implies encompassing.  The other
direction of the implication, however, does not hold: if
$\mathcal{M}_{W}$ encompasses $\mathcal{M}_{X}$, then it is not
necessarily true that the covariates $X$ are not significant in the
nonparametric regression of $Y$ onto $(W,X)$, as shown below.
\begin{prop}
 $\mathcal{M}_{W}$ encompasses $\mathcal{M}_{X}$ if and only if
 \[
\E \left( Y | W, X \right) = \E (Y|W) + g(W,X)
\qquad \mbox{with } \, \E\left[ g(W,X) | W \right] = \E\left[ g(W,X) | X \right] = 0
\as
%\, .
\]
\end{prop}
\begin{proof}
%We first prove necessity.
By definition, $Y = \E (Y|W) +\varepsilon$ with $\E (\varepsilon |W) =
0$. Hence,
\[
\E (Y|W,X)  = \E (Y|W) + \E(\varepsilon|W,X) = \E (Y|W) + g(W,X)
\, .
\]
Conditioning on $W$ yields $\E \left[ g(W,X) | W \right]
 = 0$.  Conditioning on $X$ and using the fact that
 $\mathcal{M}_{W}$ encompasses $\mathcal{M}_{X}$ yields $\E \left[
 g(W,X) | X \right] = 0$. This shows necessity.
Sufficiency similarly follows by conditioning
$\E \left( Y | W, X \right) = \E (Y|W) + g(W,X)$ on $X$
 and using $\E \left[ g(W,X) | X \right] = 0$.
\end{proof}
The previous result highlights that  for non nested sets of regressors,
 the encompassing property does not reduce to the significance of some
 regressors in the complete regression.  As long as $g(W,X)$ is not
a.s. equal to zero, $\mathcal{M}_{W}$ encompasses $\mathcal{M}_{X}$,
but $X$ is a significant covariate in $\E(Y|W,X)$.
The following  provides a concrete example.

\paragraph{Example.}
Let $W$ be continuous univariate, symmetrically distributed around 0,
with density $f_W(\cdot)$. Consider $p(\cdot) \in (0,1)$, $\varphi(\cdot)$
and $\psi(\cdot)$ two univariate densities with mean 0. Define
\[
f_{X|W}(x|w)  = p(w) \varphi(x) + (1-p(w)) \psi(x)
\, ,
\]
so that $X$ has a mixture distribution conditionally on $W$.
Then
\begin{align*}
f_{X}(x) &  =
\left[ \int p(w) f_W(w) \,  dw \right] \varphi(x) +
\left[ \int (1-p(w)) f_W(w) \,  dw \right]
 \psi(x)
 \, ,
% = a_1 \varphi(x) + a_2 \psi(x)
 \\
f_{W|X}(w|x) & =
  \frac{f_W(w)}{f_X(x)} \left( p(w) \varphi(x) + (1-p(w)) \psi(x)\right)
\, .
\end{align*}
Let
\[
Y = m(W) + h(W) X + \eta\, , \qquad  \E(\eta|W,X) = 0
\, .
\]
Then
\[
\E(X|W) =
p(W)  \int_{}^{}{ x \varphi(x) \, dx} + (1-p(W)) \int_{}^{}{x \psi(x) \, dx} = 0
\, .
\]
One can select $h(\cdot)$ and $p(\cdot)$ so that $\E [h(W)|X]
= 0$.  Consider first a case where $h(\cdot)$ is an odd function with
mean 0 and $p(\cdot)$ is even. Then by the symmetry of $f_W(\cdot)$,
$\E [h(W)|X] = 0$.  Another setup used in our simulations below is
for $h(\cdot)$ even with mean 0 and $p(-w) = 1-p(w)$. Again, it is
easy to show that  $\E [h(W)|X] = 0$.

Whenever  $\E [h(W)|X] = 0$, then $\E [ h(W) X|W] = \E[ h(W)
X|X]=0$, and $\mathcal{M}_{W}$ encompasses 
$\mathcal{M}_{X}$, but $X$ is significant in $\E(Y|W,X)$.
Furthermore, $\mathcal{M}_{X}$ does not encompass $\mathcal{M}_{W}$ in
general. If this holds, then  $m(W)=\E(Y|X)$ \as \ from
Proposition \ref{prop: Better explanation by encompassing}. Integrating both sides
with respect to the marginal density $f_X(\cdot)$ implies that $m(W)$
and thus  $\E(Y|X)$ are both constant almost surely.

\section{Tests Statistics}
\label{sec:Tests}

\subsection{ICM Statistic}\label{sec:ICM Statistic}
We want to test
\[
{H}_{0} \, : \, \E \left( \E(Y|W) | X\right) = \E(Y|X) \as
\Leftrightarrow \E \left( Y- \E(Y|W) | X\right) = 0 \as
\]
 against its logical complement ${H}_{1} =
{H}_{0}^{c}$.  While $H_0$ is a conditional moment restriction, we can consider instead an equivalent continuum of unconditional moments.
Assume $X \in \Real^d$ has bounded support, which is without loss of
generality, as we can always transform  $X$ by a
one-to-one function that maps it to a compact set.
Then the null hypothesis is equivalent to
\begin{equation}
{H}_{0} \, : \,
\E \left[ \left( Y - \E(Y|W) \right) \varphi(s'X) \right] = 0
\ \as \quad \forall s
\in {\cal S}
\, ,
\label{eq:Moment Condition}
\end{equation}
where  ${\cal S}$ is a (arbitrary) neighborhood  of the origin  in
$\Real^d$ and $\varphi(\cdot)$ is a well-chosen function, see
Assumption \ref{Ass0} below for precise conditions.
Some convenient choices for $\varphi(\cdot)$ are as follows.
\citet{bierens_consistent_1982} shows the previous equivalence for  the
complex exponential $\varphi(u) =\exp (i u) $,
\cite{bierens_consistent_1990} considers the exponential $\varphi(u) =\exp (u) $,
\citet{bierens_asymptotic_1997}  the logistic
c.d.f. $\varphi(u) = 1/(1+\exp(c-u))$, see also
\citet{stinchcombe_consistent_1998}. Other types of functions could be
used such as indicator functions \citep{escanciano_consistent_2006, delgado_significance_2001}.

If we observe a random sample $(Y_i,W_i,X_i)$, $i=1, \ldots n$, from
$(Y,W,X)$, and if we know the precise form of $\E(Y|W)$, then
Bierens' Integrated Conditional Moment (ICM)  statistic for testing $H_0$ is
\begin{equation}
 \int_{\Real^{d}}{\left| n^{-1/2} \sum_{j=1}^{n}{
 \left(Y_j -  {\E} (Y | W_j) \right) \varphi(s'X_j)} \right|^{2}
 \,  d\mu (s)}  \, ,
\label{icm}
\end{equation}
where $\mu$ is some probability measure on ${\cal S}$, such as
the uniform distribution on ${\cal S}$. Alternatively, a
Kolmogorov-Smirnov type statistic could be considered, but the
Cramer-von-Mises form appears to be easier to deal with in practice, see below.
%as done by \cite{bierens_consistent_1982}.

In practice, we use a kernel nonparametric estimator of the
conditional expectation $\E (Y|W)$.  Let $K(\cdot)$ be a kernel on
$\Real^{p}$, $h\,(=h_n)$ a bandwidth, and $K_{{h}} (u) = K \left( u_{1}/h,
\ldots u_{p}/h \right)$.\footnote{For notational simplicity we are
assuming the same bandwidth across regressors, but our proofs would carry
over when each regressor has a specific bandwidth.}  Define
\begin{equation*}
\overline{Y}(w) =  \left(n  h^p \right)^{-1}
\sum_{i=1}^n Y_i K_h \left( {W_i - w} \right)
\, .
\end{equation*}
With $e = (1, \ldots, 1)'$, let $\widehat{f} (w) = \overline{e} (w)$,
and $\widehat{m} (w) = \overline{Y}(w)/\widehat{f} (w)$.

To control for the random denominator in the kernel estimator, we
introduce a trimming factor $\widehat{t}(w) = \1 \left( \widehat{f}(w)
\geq \tau_{n} \right)$, where $\tau_n$ converges to zero.  Let
$\widehat{\varepsilon}_{i} = Y_i - \widehat{m} (W_i) $, $\phi_s(\cdot)
= \varphi(s'\cdot)$, and $\mathbb{P}_{n}(g) = n^{-1}
\sum_{i=1}^{n}g(Z_{i})$ denote the  empirical mean process based on $g(\cdot)$.
An ICM test statistic can be  built  upon the
process $\mathbb{P}_{n}\left( \widehat{\varepsilon} \phi_{s} \widehat{t} \right)$ as
\[
S_n  = n \,  \int_{{\cal S}} | \mathbb{P}_{n} \left(
\widehat{\varepsilon} \phi_{s}   \widehat{t} \right) |^{2}
\, d\mu(s)
\, .
\]
In our practical implementation, we chose $\varphi(\cdot)$ as the complex
exponential and $\mu$ to be symmetric around the origin. Define
\[
a(z) = \int_{{\cal S}} \exp(i s'z) \, d\mu(s) =
\int_{{\cal T}} \cos(s'z) \, d\mu(s)
\, ,
\]
due to the symmetry of $\mu$, and let $\widehat{\bm{  \varepsilon}} = \left(
\widehat{\varepsilon}_{j}, j=1, \ldots n \right)'$ and
$\widehat{t}_{i} = \widehat{t} (W_i)$. Then the statistic becomes
 \begin{align*}
&  \int_{{\cal S}}{ n^{-1} \sum_{j=1}^{n}\sum_{m=1}^{n}
{ \widehat{\varepsilon}_{j}  \widehat{\varepsilon}_{m}
 \widehat{t}_j  \widehat{t}_m
\exp(i s'(X_{j}-X_m))} \, d\mu (s)}
 \\
& =
 n^{-1}  \sum_{j=1}^{n} \sum_{m=1}^{n} { \widehat{\varepsilon}_{j}
 \widehat{\varepsilon}_{m}  \widehat{t}_j  \widehat{t}_m
 \int_{{\cal S}} { \cos( s'(X_{j}-X_m))} \, d\mu (s)}
=
 \widehat{\bm{\varepsilon}}^{ \prime}  \bm{A} \widehat{\bm{\varepsilon}}
\, ,
 \end{align*} where $\bm{A}$ is a matrix with generic element $n^{-1}
a\left(X_{j}-X_{m}\right) \widehat{t}_j \widehat{t}_m$.  In practice
the function $a(\cdot)$ is the Fourier transform of $\mu$, so we can
choose the latter so that the former has an analytic expression, and
computation of the matrix $A$ is fast. To achieve scale invariance, we
recommend, as in \cite{bierens_consistent_1982}, to scale each
component of $X$ by a measure of dispersion, such as the empirical
standard deviation.

The behavior of $\sqrt{n} \mathbb{P}_{n}\left( \widehat{\varepsilon} \phi_{s}
\widehat{t} \right)$ is studied in detail by
\cite{escanciano_uniform_2014}, who derived a uniform expansion.
Hence, the properties of the ICM test based on $S_n$ can be derived
 from their results.  However, these impose undersmoothing in kernel
 estimation, which ensures that the bias disappears fast enough, but
 makes the practical bandwidth choice tricky. In what follows, we
 develop two approaches to avoid undersmoothing, and in particular to
 allow for an optimal nonparametric bandwidth. This is convenient
 because there are well-known methods for approximating such
 bandwidths.

\subsection{Bias Corrected Estimation}

\cite{newey_twicing_2004} obtained a small bias property for density weighted average semiparametric estimators.  Here we instead rely on an idea developed in the boosting literature for kernel regressions \citep{di_marzio_boosting_2008, park_l2_2009}.
\citet{xia_goodness--fit_2004} used a similar bias correction in a
 specification test for a single-index model. 
%\footnote{See also \citet{buhlmann_boosting_2003} for an analysis of
%the boosting algorithm for %regression estimation with smoothing
%splines.}

$L^2$ boosting starts with an initial nonparametric estimator and
builds a bias-corrected updated estimator. The bias correction is
roughly based on nonparametric residuals. The method can be iterated,
but we will restrict to a single boosting step, which is sufficient
for our purpose. In our context, the bias correction to be applied to
$\widehat{m} (w)$ is
\[
\widehat{B} (w) = \frac{\overline{\widehat{m}}(w)}{\widehat{f} (w)} -
\widehat{m}(w)
\, ,
\]
where
\begin{equation}
\label{smoothing of m.hat}
\overline{\widehat{m}}(w)=(nh^p)^{-1}\sum_{i=1}^n\widehat{m}(W_i)
\widehat{t}_{i} K_h(W_i-w)
\, ,
\end{equation}
and the trimming  controls for the random denominator of $\widehat{m}$.
The bias-corrected estimator thus is
\begin{equation}
\widetilde{m}(w) = \widehat{m}(w)-\widehat{B}(w) =
2 \widehat{m}(w) - \frac{\overline{\widehat{m}}(w)}{\widehat{f} (w)}
\, .
\label{eq:Bce}
\end{equation}
We  then consider $\widetilde{\varepsilon}_{i} =
Y_{i} - \widetilde{m} (W_i)$ and the bias-corrected ICM statistic
\[
S_{n}^{BC}= n\,
\int_{{\cal S}} | \mathbb{P}_{n} \left(
\widetilde{\varepsilon} \phi_{s} \widehat{t} \right) |^{2}
\, d\mu(s)
\, .
\]
The form (\ref{eq:Bce}) of our bias corrected estimator is similar to
the one discussed in \citet[Section 3]{newey_twicing_2004} who
considered density-weighted nonparametric estimators.  To use their
technique in our setup would necessitate applying their correction to
both the numerator and the denominator of the regression estimator.
We feel more natural and practically more convenient to correct for
the bias of the regression estimator as described above.

\subsection{Locally Robust Process}

Locally robust semiparametric estimation has been considered by
several authors, see \cite{newey_semiparametric_1990} for an early
example.  \cite{Chernozhukov2022} consider a general GMM
estimation problem, where the moments depend on a first-step
nonparametric estimator.  Locally robust semiparametric GMM estimators
are built from moment conditions that have zero derivatives with
respect to the first-step estimator so that the latter does not affect
the asymptotic variance of the parameters of interest.  They have
smaller bias and better small sample properties than standard GMM
estimators.

In our case, we have a continuum of moment conditions.  Our goal is to
modify these moment conditions so that (i) we can still test for our
null hypothesis of interest using these modified moments, and (ii)
estimation becomes “adaptive” with respect to the nonparametric
regression. Hence, we consider the moments
\[
\E \left[ \left( Y - \E(Y|W) \right) \left( \varphi(s'X) - \E (\varphi(s'X) | W) \right)
  \right] = 0 \ \quad \forall s \in {\cal S}
  \, .
\]
Because for any $s$, $\E \left[ \left( Y - \E(Y|W) \right)  \E
(\varphi(s'X) | W)   \right] = 0$,
the above moment conditions are equivalent to (\ref{eq:Moment Condition}).
We show below that for the empirical equivalent of these moments,
estimation of nonparametric components has no first order effect.

Let $\iota_{s}(W) = \E (\phi_{s}(X)|W)$ and its estimator
\[
\widehat{\iota}_s (w) = \overline{\phi_{s}}(w)/\widehat{f} (w)
\, , \qquad
 \overline{\phi_{s}}(w)  = \left(n  h^p \right)^{-1} \sum_{i=1}^n
 \phi_s(X_i) K_h \left( {W_i - w} \right)
 \, .
\]
Our locally robust approach thus delivers the statistic
\[
S_{n}^{LR}=
n \, \int_{{\cal S}} | \mathbb{P}_{n} \left(
\widehat{\varepsilon} ( \phi_{s} - \widehat{\iota}_s)
\widehat{t} \right) |^{2}
\, d\mu(s)
\, .
\]
\cite{newey_twicing_2004} note that their bias correction
based on twicing kernels is equivalent to a locally robust density
weighted average. By contrast, our bias correction detailed in the
previous section does not yield the same statistic as the one based on the
locally robust process.
In practice, there is no need to compute  $\widehat{\iota}_{s}(W)$ for each $s
 \in {\cal S}$ and to integrate. Simple algebra reveals that, if
 $\bm{K}$ denotes the $n\times n$ matrix of generic element
 $K_h(W_i-W_j)/[nh^p\widehat{f}(W_j)]$, then
\[
S_n^{LR} = \widehat{\bm{\varepsilon}}' \bm{A} \widehat{\bm{\varepsilon}}
- 2\widehat{\bm{\varepsilon}}' \bm{A} \bm{K} \widehat{\bm{\varepsilon}}
+ \widehat{\bm{\varepsilon}}' \bm{K}' \bm{A} \bm{K}\widehat{\bm{\varepsilon}}
\, , \]
so  the locally robust version of the statistic is practically straightforward to compute.

\section{Asymptotic Analysis}
\label{sec:Theory}

We here focus on the asymptotic expansion of the empirical processes
on which our test statistics $S_{n}^{BC}$ and $S_{n}^{LR}$ are based.
We do not formally consider the empirical process entering $S_{n}$:
its properties would be similar but would necessitate assuming some
undersmoothing.

We first introduce some definitions. Define the differential operator
%of a function $g(w)$ as
\[
\partial^{l}g(w)  =
\frac{\partial^{|l|}}{\partial^{l_{1}}w_{1} \ldots
\partial^{l_{p}}w_{p}} g(w)
\qquad l=(l_{1},..,l_{p})', \ |l| = l_{1}+..+l_{p}
\, .
\]
\begin{defn}
(a)
$\mathcal{G}_\lambda(\mathcal{A}) = \left\{g:\mathcal{A}\mapsto\mathbb{R} : \quad \sup_{a\in\mathcal{A}}|\partial^l g(a)|< M \text{ for all } |l|\leq \lambda \right\}$.
(b) $\mathcal{K}_{\lambda}^{r}$ is the class of product univariate
kernels $k(\cdot)$ such that $k(\cdot)$ is of order $r$, $\lambda$
times continuously differentiable with uniformly bounded derivatives,
symmetric about zero, and with bounded support.
\end{defn}

\begin{hp}
\label{Ass0}
(i) $(Y_i, W_i, X_i), i = 1, \ldots n$, is a random sample from $(Y,W,X)$. $\mathcal{Y}\subset \Real$,
${\cal W} \subset \Real^p$, and ${\cal X} \subset \Real^{d}$, the
supports of $Y$, $W$, and $X$, are bounded.
(ii) ${\cal S}$ is a bounded compact subset of $\Real^d$ containing a
neighborhood of the origin.
(iii) $\varphi(\cdot)$ is an analytic non polynomial function
with $\partial^l\varphi(0)\neq0$ for all $l\in\mathbb{N}$.
\end{hp}

\begin{hp}
\label{Assumption: 1}
(i) $f(\cdot)\in \mathcal{G}_{r}({\Real^p})$ and $m(\cdot)\in
 \mathcal{G}_{r}(\cal W)$, with $r \geq   \lceil (p+1)/2
 \rceil$.\footnote{$\lceil x \rceil$ denotes the smallest integer
 above $x$.}
(ii) $\iota_s (\cdot)\in \mathcal{G}_{r}({\cal W})$ uniformly in $s
 \in {\cal S}$, with $r \geq  \lceil (p+1)/2 \rceil$.
 (iii) $K(\cdot)\in \mathcal{K}_{\lambda}^{r}$ with $\lambda \geq
 \lceil (p+1)/2 \rceil $.
\end{hp}

Define the uniform convergence rate of the kernel
density estimator as
\begin{align*}
 d_{n}&  = \sqrt{\frac{\log n}{n h_{n}^{p}}} + h_{n}^{r}
 \, .
\end{align*}
\begin{hp}
\label{Assumption: 3}
(i) $h_n \tau_{n}^{-1} = o(1)$. (ii)
$d_{n}\tau_{n}^{-3}=o(n^{-1/4})$. (iii)
$d_{n}h^{-|l|} \tau_{n}^{-(2+|l|)} = o(1)$ for $|l| =  \lceil (p+1)/2 \rceil$.
\end{hp}

\begin{hp}
\label{Assumption: 3bis}
(i) $p_{n}  = \Pr \left( f(W) \leq 3\tau_{n}/2 \right)$
is such that $p_{n}=o(n^{-1/2})$ ,
$p_{n} h_{n}^{-p}\tau_{n}^{-2}=o(n^{-1/4})$, and
$p_{n} h_{n}^{-(p+|l|)}\tau_{n}^{-(1+|l|)}=o(1)$
for $|l| =  \lceil (p+1)/2 \rceil$.

(ii) There exists $N$  such that for all $n\geq N$,
$\mathcal{W}_{n}=  \left\{ w : f(w) \geq
\tau_{n}/2 \right\}$ is convex.
\end{hp}

Assumption \ref{Ass0} ensures that the encompassing hypothesis can be
written as the continuum of moment conditions (\ref{eq:Moment
Condition}).  In particular, \citet[Theorem
2.2]{bierens_econometric_2017} builds on previous results by
\citet{bierens_consistent_1982} and
\citet{stinchcombe_consistent_1998}, and shows that
\ref{Ass0}-(iii) is sufficient. Intuitively, if $\varphi(\cdot)$ is analytic
non-polynomial, Equation (\ref{eq:Moment Condition}) implies that $Y -
\E(Y|W)$ is uncorrelated with any polynomial in $s'X$, and thus yields
the required equivalence.  This allows for different functions
$\varphi(\cdot)$ as previously detailed.

Assumption \ref{Assumption: 1} imposes conditions that are commonly
found in the literature on nonparametric estimation. In particular,
they impose that the density of $W$ is differentiable over $\cal W$
and its derivatives go smoothly to zero as we approach to the
boundaries of the support. Part (ii) is necessary only for the study
of the locally robust empirical process, which involves a
nonparametric estimator of  $\iota_s (\cdot)$.
%Assumption  \ref{Assumption: 1}-(ii) could likely be weakened, and larger classes of
%functions could be considered, including classes of functions with
%discontinuities, see e.g. \cite{AndrewsHB}.  For the ease of exposition,
%we chose to assume the same kind of conditions hold for the several
%classes of functions that we consider.

Assumption \ref{Assumption: 3} sets the main conditions on the
bandwidths.  Together with Assumption \ref{Assumption: 1}, it implies
that the kernel estimators asymptotically belong to a class of
sufficiently smooth functions with limited entropy/complexity. This is
needed for the asymptotic stochastic equicontinuity of the empirical
processes at the basis of our test statistics. Assumptions
\ref{Assumption: 3} and \ref{Assumption: 3bis} impose conditions on
the bandwidths in connection with the trimming parameter. We need the
effect of trimming to vanish quickly enough to avoid large biases.
Abstracting from the appearance of the trimming, Condition (ii) in
Assumption \ref{Assumption: 3}  ensures
that the kernel estimators are $n^{-1/4}$-consistent. It requires in
particular that $n h_{n}^{4r}=o(1)$.
Without the bias correction or locally robust
approach employed in the construction of our statistics,
this condition would become $n h_{n}^{2r}=o(1)$
to ensure that the  bias of the nonparametric estimator is negligible
compared to the variability of the empirical process,
see  \citet{delgado_significance_2001} or
\citet{escanciano_uniform_2014}.
This means that our empirical processes have a small bias property (in
their stochastic expansion).  Our conditions allow the bandwidth to be
optimal for nonparametric estimation purposes and avoid the need for
undersmoothing.  The main practical advantage is that one can use for
instance a simple rule-of-thumb.\footnote{We note however that the
optimal nonparametric bandwidth is likely not optimal for estimation
of $\E \left[ \left( Y -\E(Y|W) \right) \varphi(s'X) \right]$.}  We
could extend our theoretical results to stochastic bandwidths to allow
for the use of data-driven methods, as was done in other contexts
\citep{andrews_nonparametric_1995, Mammen1992, Lav2001}. The details
would be more involved here because of the stochastic trimming, so we
do not pursue this issue further.

Assumption \ref{Assumption: 3bis} as a whole is used to ensure that
trimming has a negligible effect overall, see
e.g. \cite{LV96} or \cite{escanciano_uniform_2014} for similar assumptions.  It may be difficult to check, but our
simulations seem to indicate that trimming is actually not
crucial: when a data-driven bandwidth is used, results with no
trimming are comparable, and mostly better, than results with
trimming.

The following proposition establishes the influence function
representations of the empirical processes used in our statistics
$S_{n}^{BC}$ and $S_{n}^{LR}$.
\begin{prop}
\label{prop: Asymptotic Test}
Under Assumptions \ref{Ass0}--\ref{Assumption: 3bis},
$ \sqrt{n}\mathbb{P}_{n} (\widetilde{\varepsilon}\phi_{s} \widehat{t})
=\sqrt{n}\mathbb{P}_{n}(\varepsilon (\phi_{s} - \iota_s)) + o_{P}(1)$ and
$
\sqrt{n}\mathbb{P}_{n} (\widehat{\varepsilon}(\phi_{s}-\widehat{\iota}_s) \widehat{t})
= \sqrt{n}\mathbb{P}_{n}(\varepsilon (\phi_{s} - \iota_s)) + o_{P}(1)$
uniformly in $s \in {\cal S}$.
\end{prop}
From the above result, under $H_0$ the limiting distribution, say $\mathbb{G}_s$, of both empirical
processes is the limiting one of $\sqrt{n}\mathbb{P}_{n}(\varepsilon (\phi_{s}
- \iota_s))$. This $\mathbb{G}_s$
is a zero-mean tight Gaussian process valued in $L^\infty(\cal S)$,
the space of uniformly bounded  functionals over $\cal S$, and
characterized by the collection of covariances
$\{\E (Y-m(W))^2 (\phi_s(X)-\iota_s(X))
(\phi_t(X)-\iota_t(W)) :s,t \in\mathcal{S}\}$.
Our asymptotic expansions directly yield
\[
S_{n}^{BC} \quad \mbox{or } \quad S_{n}^{LR} \
\cvd \int_{}^{}{\left|\mathbb{G}_s \right|^2 \, d\mu(s)}
\, .
\]
 %The above limit distribution is continuous, as can be deduced from
 %the results of \cite{bentkus_asymptotic_1993}. \textcolor{magenta}{(Do you think we should drop the %previous sentence?) }
 However, the result  is not useful in practice since the covariance function of
 $\mathbb{G}_s$ depends on the unknown data generating process.
 In what follows, we develop a bootstrap procedure for obtaining
 critical and p-values.

\section{Bootstrap Tests}
\label{sec:Bootstrap}

 We use a wild bootstrap procedure that imposes the null
 hypothesis $H_{0}$ when resampling the observations. The
 bootstrap data generating process (DGP) writes as
 \begin{align*}
 Y_{i}^{*} & =\widehat{m}(W_{i}) + \xi_{i} \widehat{\varepsilon}_{i}
 \qquad \widehat{\varepsilon}_{i} = Y_{i}- \widehat{m}(W_{i})
 \, ,
%\label{eq:BootstrapDGP}
\end{align*}
where $\{\xi_{i}, i=1, \ldots n\}$ is a sequence of
independent bootstrap weights with $\E\xi=0$ and
$\E\xi^{2}=1$.
%The bootstrap DGP will be the same for both our statistics.
From each bootstrap sample $\{\left(Y_{i}^*, W_i, X_i\right)\text{ :
}i=1,..,n\}$, we proceed as above to obtain $\widehat{\varepsilon}^{*}
= Y_{i}^{*}- \widehat{m}^{*}(W_{i})$ and
$\widetilde{\varepsilon}_{i}^{*} =
Y_{i}^{*}-\widetilde{m}^{*}(W_{i})$, where
\begin{align*}
\widetilde{m}^{*}(w)=\widehat{m}^{*}(w)-\widehat{B}^*(w)
& =  \widehat{m}^{*} (w)
 -
\left(  \frac{\overline{\widehat{m}^*}(w)}{\widehat{f} (w)} -
\widehat{m}^*(w) \right),
\qquad \widehat{m}^*=\overline{Y^*}/\widehat{f}
\, .
\end{align*}
The kernel smoothers $\overline{Y^*}$ and $\overline{\widehat{m}^*}$
are constructed in the same way as in Equation (\ref{smoothing of
m.hat}) to control for the random denominators in $Y^*$ and
$\widehat{m}^*$.
The bias correction could be the one used in the
original statistic, as done by \cite{xia_goodness--fit_2004}, however
we noted in simulations that recomputing the bias correction with
bootstrap data yields a better behavior.
We thus consider the bootstrap statistics
\begin{align}
n\int_{{\cal S}}
| \mathbb{P}_{n} \left( \widetilde{\varepsilon}^{*} \phi_{s} \widehat{t} \right) |^{2} \, d\mu(s)
& \quad \mbox{and }
\quad
n\int_{{\cal S}}
| \mathbb{P}_{n} \left( \widehat{\varepsilon}^{*}
( \phi_{s} - \widehat{\iota}_s) \widehat{t} \right) |^{2}
\, d\mu(s)
\, .
\label{eq:bootstat}
\end{align}
In practice, one can compute many bootstrap statistics and obtain their
$1-\alpha$ quantiles, denoted as ${q}_{1-\alpha}^{BC}$ and
${q}_{1-\alpha}^{LR}$.
The bootstrap test rejects $H_0$ whenever  $S_{n}^{j} > {q}_{1-\alpha}^{j}$, for
$j=BC$ or $LR$.

\begin{prop}
\label{prop: The Bootstrap Test}
Let $\{\xi_{i}, i=1, \ldots n\}$ be an i.i.d. sequence  with $\E\xi=0$,
$\E\xi^{2}=1$, and $\E|\xi|^{3}< \infty$ independent of the sample.
Under Assumptions \ref{Ass0}-\ref{Assumption: 3bis},

(a)
$ \sqrt{n}\mathbb{P}_{n} (\widetilde{\varepsilon}^{*}\phi_{s} \widehat{t})
=\sqrt{n}\mathbb{P}_{n}(\xi \varepsilon (\phi_{s} - \iota_s)) + o_{p}(1)$ and
$
\sqrt{n}\mathbb{P}_{n} (\widehat{\varepsilon}^*(\phi_{s}-\widehat{\iota}_s) \widehat{t})
=\sqrt{n} \mathbb{P}_{n}(\xi \varepsilon (\phi_{s} - \iota_s)) + o_{p}(1)$
uniformly in $s \in {\cal S}$.\footnote{Here the probability space is the joint probability on random bootstrap weights and sample data.}

(b)
Under $H_0$,
$
%\lim\sup_{n \rightarrow \infty}
\Pr \left[ S_{n}^{j} > {q}_{1-\alpha}^{j} \right] \rightarrow \alpha
$, $j=BC$ or $LR$.

(c) Under $H_1$,
$\Pr \left[ S_{n}^{j} > {q}_{1-\alpha}^{j} \right] \rightarrow 1$,
$j=BC$ or $LR$.
\end{prop}

\section{Numerical Results}
\label{sec:Simulations}

\subsection{Small Sample Behavior}

We used a DGP in line with our example in  Section \ref{sec:EncSignif}.
Specifically,
\[
Y = m(W) + h(W) X + \gamma \delta(X) + \eta\, , \qquad  \E(\eta|W,X) = 0
\, ,
\]
where $\eta \sim N(0,1)$ is independent of $(W,X)$, $W \sim
N(0,\sigma^{2} = (1/2)^2)$,
\[
f_{X|W}(x|w)  = p(w) \varphi(x; 1/4) + (1-p(w)) \varphi(x; 3/4)
\, ,
\]
with $\varphi(\cdot; s)$  the density of a $N(0, \sigma^2 =s^2)$ and
$p(\cdot)$  the cumulative distribution function of $W$. We set
$m(w)  = w^3 - 2 w$,
$h(w)  = w^2 - (1/2)^2$,
$\delta(x) = x^4 - 3 x^2$.
When $\gamma = 0$, ${\cal M}_W$ encompasses ${\cal M}_X$, while it
does not when $\gamma \neq 0$.

To study the behavior of our tests based on $S_n^{BC}$ and $S_n^{LR}$
and compare them to the test based on the uncorrected ICM statistic
$S_n$, we ran 10000 simulations for different values of $\gamma$ and
for sample sizes $n=$200 and 400.  For our implementation, the
weighting function was $a(x) = \operatorname{sinc}(\pi x)$, which
corresponds to a uniform $\mu$ with complex exponential function
$\varphi(\cdot)$.  The observations on $X$ were transformed by the
logistic c.d.f. then standardized before being passed as arguments to
$a(\cdot)$.  As there is no clear way to choose the amount of
trimming, we trimmed the 2\% more extreme observations in a first
step, and we subsequently investigated the influence of trimming.  For
bootstrapping, we used the two-point distribution defined through $
\Pr ( \xi = \frac{3-\sqrt{5}}{2}) = \frac{5 + \sqrt{5}}{10}$ and $\Pr(
\xi = \frac{3+\sqrt{5}}{2}) = \frac{5 - \sqrt{5}}{10}$.  We chose this
simple distribution with third central moment equal to one in the hope
to better approximate the distribution of the statistic, as is the
case in simpler setups \citep{Mammen1992}.     To speed up computations, we used
the warp-speed method proposed by
\cite{davidson_improving_2007} and studied by
\cite{giacomini_warp-speed_2013}. Specifically, we drew one bootstrap
sample for each simulated data, and we used the whole set of bootstrap
statistics to compute the bootstrap p-values associated with each
original statistic.

We employed Gaussian kernels of order 2 for nonparametric
estimation. In a first step, we used the bandwidth rule
$h=C\widehat{\sigma}_W n^{-1/5}$, where $\widehat{\sigma}_{W}$ is the
estimated standard deviation of $W$.  To check the performance of our
tests under different bandwidth choices, we let the constant $C$ vary.
We report actual rejection probabilities for 5\% and 10\% nominal
sizes in Table
\ref{tab:erp}. For $C=.5$ or $C=1$, the three tests have  equal size
control and the empirical size becomes closer to the nominal level
when increasing the sample size.  When $C$ increases to 1.5 then 2,
size control deteriorates for the uncorrected ICM test, while it improves for our two tests.
We also report in Figure \ref{Fig o} errors in rejection probability (ERP), that is
the difference between the empirical rejection proportion and the
nominal size under the null hypothesis $H_0$.  A perfect test would
exhibit an ERP of zero for any nominal size. This gives us a visual
way to evaluate whether the null distribution of the test statistic is
well approximated by its bootstrap approximation. These graphs clearly
show the high sensitivity of the ICM test to the bandwidth choice, and
the relative robustness of our two procedures. We complemented our study with an
investigation of the trimming influence.  While we do not report
details, the most striking result we obtained was that the absence of
trimming adversely affected the behavior of the locally robust test for large
bandwidths.

Following the suggestion of a referee, we investigated the issue of bandwidth
selection.  Since there is no role for the bandwidth in the first
order expansion of our test statistics, we decided to implement a
data-driven bandwidth that should be asymptotically optimal for
nonparametric estimation. We used an improved version of the Akaike
Information Criterion proposed by \citet{hurvich1998smoothing}, which
has been found to have excellent small sample performances by
\cite{Li-cross-2004}. The criterion is
\begin{equation*}
    AIC_c=\log\left(\frac{1}{n}\sum_{i=1}^n [Y-\widehat m (W_i)]^2\right)
    +\frac{1+tr(H)/n}{1-[tr(H)+2]/n}\, ,
\end{equation*}
where $H$ is the “hat matrix'' such that $\left( \widehat m
(W_1),\ldots,\widehat m(W_n)\right)'= H
\,\left(Y_1,\ldots,Y_n\right)'$.  As seen from Table \ref{tab:erp},
with this data-driven bandwidth, the three  tests have good size
control.
We also report
errors in rejection probability in Figure \ref{Fig i}, which shows the
excellent adequation of the bootstrap approximation. We then repeated
our experiment without any trimming and observed that ERP are closer
to zero for our tests compared to the case where trimming is
implemented.

Finally, for a data-driven bandwidth and no trimming, we report in
Figure \ref{Fig ii} the power of the three tests when $\gamma$
varies. Our two tests basically have the same power, and power
increases with sample size, while the uncorrected ICM test is
significantly less powerful.

\begin{table}
\centering
\begin{tabular}{l | r  r  r  r |  r  r}
& \multicolumn{4}{c|}{ $h = C \widehat{\sigma}_W n^{-1/5}$ } & \multicolumn{2}{c}{ Data-driven $h$ } \\
\hline
& $C=0.5$ & $C=1$ & $C=1.5$ & $C=2$ &
Trimming & No trimming \\
\hline
$n=200$  & & & & & & \\
\hline
Bias Corrected &
12.15 & 11.13 & 10.35 & 9.83 & 11.34 & 11.43\\
& 6.86 & 5.64 & 5.56 & 5.08 & 5.74 & 5.79\\

\hline
Locally Robust &
12.08 & 11.05 & 10.39 & 9.71 & 11.48 & 11.27\\
& 6.77 & 5.67 & 5.38 & 5.36 & 5.77 & 5.67\\

\hline
ICM &
12.27 & 11.09 & 11.51 & 13.51 & 11.49 & 11.31\\
& 6.65 & 5.93 & 6.02 & 7.18 & 6.05 & 6.13\\

\hline
$n=400$  & &  &  &  &  &  \\
\hline
Bias Corrected &
11.41 & 10.74 & 10.31 & 9.43 & 10.70 & 10.64\\
& 5.90 & 5.60 & 5.33 & 4.99 & 5.73 & 5.79\\
\hline
Locally Robust &
11.32 & 10.73 & 10.17 & 9.70 & 10.81 & 10.24\\
& 5.91 & 5.64 & 5.21 & 4.91 & 5.58 & 5.75\\

\hline
ICM &
11.13 & 10.19 & 10.93 & 13.24 & 10.46 & 10.42\\
& 6.03 & 5.69 & 5.71 & 6.69 & 5.85 & 5.80\\
\hline
\end{tabular}
\caption{Empirical rejection percentages under $H_0$ at 10\% and 5\% nominal levels.}
\label{tab:erp}
\end{table}

\subsection{Empirical Illustration}

We apply our encompassing tests to competing models of consumption behavior, in the spirit of  \citet{gaver1974discriminating} and \citet[Chapter 8]{greene2003econometric}.
The first model assumes that consumption only depends on income, and relates consumption to current and past income. The second assumes the presence of habits in consumption behavior, and relates consumption to current income and past consumption. Formally, the first model is
\begin{equation}
\label{eq: cons model 1}
    C_{i}=g_1(I_{i},PI_{i})+\varepsilon_{i}\, 
    \qquad \E[\varepsilon_{i}|I_{i},PI_{i}]=0\, ,
\end{equation}
where $C_{i}$ denotes  (log of) consumption of household $i$, $I_{i}$ denotes (log of) income of household $i$, and $PI_i$ denotes  past (log of) income of household $i$. 
The second model is
\begin{equation}
\label{eq: cons model 2}
     C_{i}=g_2(I_{i},PC_{i})+\eta_{i}\, 
     \qquad \E[\eta_{i}|I_{i},PC_{i}]=0\, ,
\end{equation}
where $PC_i$ denotes (log of) past consumption of household $i$. 
%Thus, while Model (\ref{eq: cons model 1}) assumes that consumption responds to changes in income over two periods,  Model (\ref{eq: cons model 2}) assumes that the changes in income have a persistent effect on consumption. 

%To select between these two competing theories, we will apply our %encompassing tests. 
We used data compiled by \citet{arellano2017earnings} from  the Panel Study of Income Dynamics, which concerns $n=792$ households. We focus on the first two time periods, namely the years $1999$ and $2001$, so that $C_i$ and $I_i$ correspond to variables in 2001, while $PC_i$ and $PI_i$ correspond to  1999.
We first tested if Model (\ref{eq: cons model 1}) encompasses Model (\ref{eq: cons model 2}), i.e.
\begin{equation}
\label{eq: empirical null hypothesis 1}
    H_0\, : \E[g_1(I_{i},PI_{i})\,|\,I_{i},PC_{i}] = g_2(I_{i},PC_{i})
    \, .
\end{equation}
We then tested the reversed hypothesis that Model (\ref{eq: cons model 2}) encompasses Model (\ref{eq: cons model 1}), i.e.,
\begin{equation}
\label{eq: empirical null hypothesis 2}
    H_0'\,:\, \E[g_2(I_{i},PC_{i})|I_{i},PI_{i}]=g_1(I_{i},PI_{i})
    \, .
\end{equation}
We applied our encompassing tests with the bandwidth selected  by the $AIC_c$ criterion and 999 bootstrap samples drawn to obtain critical values. Table \ref{table: empirics 1 and 2} reports the values of  our test statistics, together with their 95 percent bootstrap quantiles and bootstrap p-values. 
The first encompassing hypothesis $H_0$ is clearly rejected, hence Model (\ref{eq: cons model 1}) does not encompass  Model (\ref{eq: cons model 2}). The encompassing hypothesis $H_0'$ however cannot be rejected by any of our tests, even at a 10\% nominal level. This suggests that consumption may be  adequately modeled as a function of income and past consumption, in a way coherent with the consumption habit formation theory. % \textcolor{red}{(not with the "permanent income theory", but with the "habit formation assumption")}.

\begin{table}[ht]
\centering
\begin{tabular}{l | r | r | r }
 &  Statistics &  5\% Critical Value &  P-value \\
  \hline
  Test of $H_0$ &  & &  
 \\
 \hline
  Bias Corrected & 3.0848 &  0.1820 &  0.0000 \\
 Locally Robust & 3.1125 &  0.1968 &  0.0000 \\
ICM & 3.3948 &  0.2499 &  0.0000 \\
 \hline 
 {Test of $H'_0$} & & &
 \\
\hline
  Bias Corrected & 0.0069  & 0.0294  & 0.8799 \\
  Locally Robust & 0.0355  & 0.0447  & 0.1912 \\
ICM & 0.0249  & 0.0539  & 0.5485 \\
   \hline
%   \hline
\end{tabular}
%\begin{minipage}{\textwidth}
%\linespread{1.1}\selectfont
%\footnotesize
%\vspace{0.4cm}
%Notes: Statistics for the uncorrected (ICM), the bias corrected (Bias Corr.), and the locally robust (Locally Rob.) encompassing test. The quantiles and the p-values are computed by the wild bootstrap method of Section (\ref{sec:Bootstrap}). The sample size is 792 and the number of bootstrap iterations is 999.
% \vspace{0.2cm}
%\end{minipage}
\label{table: empirics 1 and 2}
\caption{Test of the encompassing hypotheses $H_0$ in (\ref{eq: empirical null hypothesis 1}) and $H_0'$ in (\ref{eq: empirical null hypothesis 2})}
\end{table}

\section{Concluding Remarks}

% \textcolor{red}{(Maybe should we say that we extended and studied the concept of encompassing in a nonparametric context?)} 
We have studied two general approaches to obtain a small bias property
for our encompassing test statistic.  These approaches could
potentially be used in other estimation and testing problems involving
a first-step nonparametric estimation. Our  simulation
experiment seems to indicate that when coupled with a data-driven
bandwidth these two approaches are  robust to trimming choices and
deliver good performances. 
%\textcolor{red}{(We should also say that these two approaches are robust to bandwidth perturbations, contrary to the uncorrected ICM test.. this would highlight once more the externt of our contributions)}

\section{Proofs}
\label{sec:Proofs}

%\textbf{Notation}.
For any real or complex valued function $g(\cdot)$, we denote with  $||g||_\infty$ the supremum norm  taken over the support of its argument and
$P g  = \int g(z) \, dP(z)$.
We  define the population counterpart of $\widehat{t}$ as
$t(w)=\1(f(w)\geq \tau_n)$.
By convention, $(\overline{Y}/\widehat{f})(w)$ will be set to 0 whenever $\widehat{f}(w)=0$. The same notation will hold for
$\overline{\widehat{m}}/\widehat{f}$, $\overline{\phi}_s/\widehat{f}$,  $\overline{Y^*}/\widehat{f}$, and $\overline{\widehat{m}^*}/\widehat{f}$.
$C$ denotes a generic constant that may vary from line to line.

\subsection{Proof of Proposition \ref{prop: Asymptotic Test}}
%\textcolor{magenta}{Revision of Proposition 1: Done}\\
(a) We  study  the bias corrected  empirical process
\begin{align}
 \sqrt{n}\mathbb{P}_{n} (\widetilde{\varepsilon}\phi_{s} \widehat{t})
& =
\sqrt{n}\mathbb{P}_{n} \left(  \left( Y - \widehat{m} + \widehat{B}
\right) \phi_{s} \widehat{t} \right)
\nonumber \\
&=
\sqrt{n}\mathbb{P}_{n} (\varepsilon \phi_{s} \widehat{t})
+
\sqrt{n}\mathbb{P}_{n} \left(  (m -  \widehat{m}) \phi_{s} \widehat{t} \right)
+
\sqrt{n}\mathbb{P}_{n} ( \widehat{B} \phi_{s} \widehat{t})
\, .
\label{eq:align1}
\end{align}
Since $|\varepsilon|$ is bounded and $|\phi_{s}(\cdot)|$ is
uniformly bounded,
\[
\sup_s \left| \sqrt{n}\mathbb{P}_{n} (\varepsilon \phi_{s} (\widehat{t} - 1)
) \right|
\leq C \sqrt{n}\mathbb{P}_{n} |\widehat{t} - 1|
= o_p(1)
\]
by Lemma \ref{lem:littleo}-(i).
Hence,   $\sqrt{n}\mathbb{P}_{n} (\varepsilon \phi_{s} \widehat{t})
 =\sqrt{n}\mathbb{P}_{n} (\varepsilon \phi_{s} )
+ o_{p}^{\mathcal{S}}(1)$,
 where $o_{p}^{\mathcal{S}}(1)$ denotes a uniform in $s \in {\cal S}$ $o_p(1)$.
For the second term, as  $(\widehat{t}-t)=
\widehat{t}^{2}-t^{2}=(\widehat{t}+t)(\widehat{t}-t)$,
and  $ \| (m - \widehat{m}) (\widehat{t} +t) \|_{\infty} = o_{p}
(n^{-1/4})$ by Lemma \ref{lem:littleo}-(ii),
\[
\sup_s  \Big| \sqrt{n} \mathbb{P}_{n} (m -  \widehat{m}) \phi_{s} (\widehat{t}-t) \Big|
  \leq  C
 \| (m -  \widehat{m}) (\widehat{t} + t) \|_{\infty}
 \sqrt{n} \mathbb{P}_{n} |\widehat{t}-t| = o_p(1)
\, .
\]

We  now  show that the  third  term is such that
\[
\sqrt{n}\mathbb{P}_{n} ( \widehat{B} \phi_{s} \widehat{t})
 =
- \sqrt{n}\mathbb{P}_{n} (\varepsilon \iota_s )
- \sqrt{n}\mathbb{P}_{n} \left(  (m -  \widehat{m}) \iota_s {t}
\right)
+o_{p}^{\mathcal{S}}(1)
\, .
\]
First, use $\|\widehat{B} (\widehat{t} + t) \|_{\infty} =
o_{p}(n^{-1/4})$, see Lemma \ref{lem:littleo}-(iii), and similar arguments
as above to show that $\sqrt{n}\mathbb{P}_{n} ( \widehat{B} \phi_{s}
\widehat{t}) = \sqrt{n}\mathbb{P}_{n} ( \widehat{B} \phi_{s} {t})
+o_{p}^{\mathcal{S}}(1)$.
Lemma \ref{lem:Belonging Conditions}-(i)-(ii) ensures in addition that
$\what{B}\in \mathcal{G}_{l}(\mathcal{W}_{n})$ with probability
approaching one, where $l =\lceil (p+1)/2 \rceil$ and $\mathcal{W}_n :=
\left\{ w : f(w) \geq \tau_{n}/2 \right\}$.
So,    Lemma \ref{lem: ASE}-(i) yields
\begin{align}
\label{ASE for bias correction}
 \sqrt{n}\mathbb{P}_{n} (\widehat{B} t \phi_{s}) & = \sqrt{n} P
 (\widehat{B} t \phi_{s}) + o_{p}^{\mathcal{S}}(1)
 \, ,
 \\
%\mbox{ where} \quad
 P (\widehat{B} t \phi_s) & =
\int_{}^{}{\widehat{B}(w) t(w) \iota_s(w) f(w) \, dw}
=
 P (\widehat{B} t \iota_s)
%\nonumber
\, .
\end{align}
Since
$\| [\widehat{B}-(\overline{\what{m}}-\overline{Y})/ f ] t \|_{\infty}
= o_p(n^{-1/2})$  from Lemma \ref{lem:littleo}-(iv) and $\iota_s$ is uniformly bounded,
\begin{align*}
\sqrt{n}  P ( \widehat{B} t \iota_{s})
& =
\sqrt{n} \int (\overline{\widehat{m}}(w)-\overline{Y}(w)) t(w)
\iota_s(w) \, dw
+ o_p^\mathcal{S}(1)
\\
& = - \sqrt{n}\mathbb{P}_{n} \left[ (Y-\widehat{t}\widehat{m})
\int K(u) (t \iota_{s})(W+uh)\, du \right]  + o_p^{\mathcal{S}}(1)
\\
& = - \sqrt{n}\mathbb{P}_{n} \left[ \left( \varepsilon -(\what{m}-m)\what{t}
- m (\what{t} -1) \right)
\int K(u) (t \iota_{s})(W+uh)\, du \right]  + o_p^{\mathcal{S}}(1)
\, ,
\end{align*}
where the second equality follows from a change of variable.
Let us  deal with each term separately.  Since
$\varepsilon$ is bounded and  $|t(w+uh)-1|=\1\{f(w+uh)<\tau_n\}$,
\begin{align*}
\Big| \sqrt{n} \mathbb{P}_{n} \varepsilon  \int K(u) \iota_{s}(W+uh) & [t(W+uh)-1] \, du \Big|
\leq
 C  \int|K(u)| \sqrt{n} \mathbb{P}_{n} \1\left\{ f(W+uh)< \tau_{n}\right\} \, du
 \, .
\end{align*}
By a mean-value expansion, for all $(u,w)\in\text{Supp}(K) \times \mathbb{R}^p$ we have
$f(w+hu) - f(w) =  \partial^T f(\widetilde{w}) u h$,
for some $\widetilde{w}$. By Assumptions \ref{Ass0}-\ref{Assumption: 1},
$
 \left|
 \partial^T f(w) uh\right|
 \leq C h \quad \text{for all }(u,w) \in \text{Supp}(K) \times\mathbb{R}^p
$.
Since $h\tau_n^{-1}=o(1)$,  for $n$ large enough
$
 \1\{f(w+uh)<\tau_n\} \leq
 \1\{ f(w)\leq3\tau_{n}/2\}
$.
Now use $\mathbb{P}_{n} \1\Big\{ f(W)\leq3\tau_{n}/2\Big\} = O_{P}(p_{n}) = o_P(n^{-1/2})$,  from Assumption \ref{Assumption: 3bis}(i), to obtain
\[
\sqrt{n}\mathbb{P}_{n} \left[\varepsilon
\int K(u) (t \iota_{s})(W+uh)\, du \right]
=
\sqrt{n}\mathbb{P}_{n} \left[\varepsilon
\int K(u) \iota_s(W+uh)\, du  \right]
+ o_p^{\mathcal{S}}(1)
\, .
\]
The same  reasoning allows the same replacement in the second and
third term of the decomposition, using that $\|(\what{m}-m)\what{t}\|_\infty
$ is $o_p(1)$  by Lemma
\ref{lem:littleo}-(ii) and that $m(\widehat{t}-1)$ is bounded.

Then
\begin{align}
\sqrt{n}  P (\widehat{B} t \iota_{s})
& =
- \sqrt{n}\mathbb{P}_{n} \left[ \varepsilon \int K(u) \iota_s(W+uh)\,
du \right]
+ \sqrt{n}\mathbb{P}_{n} \left[ (\widehat{m}-m) \widehat{t} \int K(u)
\iota_s(W+uh)\, du \right]
\nonumber \\
&
+ \sqrt{n}\mathbb{P}_{n} \left[ m (\what{t} -1)  \int K(u)
\iota_s(W+uh)\, du \right]
+o_p^{\mathcal{S}}(1)
\, .
\label{eq:align2}
\end{align}
Lemma \ref{lem: ASE}(ii) yields
\[
\sqrt{n}\mathbb{P}_{n}\left( \varepsilon \int K(u) \iota_s(W+uh)\, du \right) =
\sqrt{n}\mathbb{P}_{n} (\varepsilon \iota_s) + o_p^{\mathcal{S}}(1)
\, ,
\]
as $P \left(\varepsilon \int K(u) \iota_s(W+uh)\, du
\right) = P \left(\varepsilon \iota_s \right) = 0
$. A Taylor expansion of order $r$ guarantees that $\int
K(u)\iota_s(w+uh)du=\iota_s(w)+O(h^r)$ uniformly in $\mathcal{S}\times
\mathcal{W}$ as $\iota_s(\cdot)$ has uniformly bounded derivatives of order $r$. Since
$\|(\what{m}-m)\what{t}\|_{\infty} = o_p(n^{-1/4})$ and $nh^{4r}=o(1)$,
\begin{align*}
\sqrt{n}\mathbb{P}_{n} \left( (\widehat{m}-m) \widehat{t} \int K(u)
\iota_s(W+uh)\, du \right)
& =
\sqrt{n}\mathbb{P}_{n}( (\widehat{m}-m)
\widehat{t} \iota_s ) +o_p^{\mathcal{S}}(1)
\, .
\end{align*}
Now use similar arguments as before to replace $\widehat{t}$ by $t$.
The last term in (\ref{eq:align2}) is negligible. Indeed,
  the argument of $\mathbb{P}_n$ is uniformly bounded, and
$\sqrt{n} \mathbb{P}_{n}|t-1|= o_P(1)$ from Lemma
\ref{lem:littleo}-(i).  Hence,
\begin{align*}
\sqrt{n} P (\widehat{B} t \iota_{s})
& =
- \sqrt{n}\mathbb{P}_{n} (\varepsilon \iota_s)
+ \sqrt{n}\mathbb{P}_{n} ( (\widehat{m}-m) {t} \iota_s)
+o_p^{\mathcal{S}}(1)
\, .
\end{align*}

Gathering results,
\[
\sqrt{n}\mathbb{P}_{n} (\widetilde{\varepsilon}\phi_{s} \widehat{t})
 =
\sqrt{n}\mathbb{P}_{n} (\varepsilon (\phi_{s} - \iota_s) )
+
\sqrt{n}\mathbb{P}_{n} \left(  (m -  \widehat{m}) (\phi_{s} - \iota_s)
{t}  \right)
+o_p^{\mathcal{S}}(1)
\, .
\]
From arguments similar to those used for
(\ref{ASE for bias correction}), since $\| (m -  \widehat{m}) t\|_\infty =
o_{p}(n^{-1/4})$ and $\what{m}\in \mathcal{G}_l(\mathcal{W}_{n})$ with probability
approaching one,
\[
\sqrt{n}\mathbb{P}_{n} \left(  (m -  \widehat{m}) (\phi_{s} - \iota_s)
{t}  \right)
=
\sqrt{n}P \left(  (m -  \widehat{m}) (\phi_{s} - \iota_s)
{t}  \right)
+o_p^{\mathcal{S}}(1)
\, .
\]
But $\sqrt{n}P \left(  (m -  \widehat{m}) (\phi_{s} - \iota_s)
{t}  \right) = 0$  as $\E (\phi_s(X)|W) = \iota_s(W)$.
Finally, we obtain
$
\sqrt{n}\mathbb{P}_{n} (\widetilde{\varepsilon}\phi_{s} \widehat{t})
 =
\sqrt{n}\mathbb{P}_{n} (\varepsilon (\phi_{s} - \iota_s) ) +o_p^{\mathcal{S}}(1)
$ as expected.

(b) We now study the locally robust empirical process
\begin{align}
 \sqrt{n}\mathbb{P}_{n} (\widehat{\varepsilon} (\phi_{s} - \what{\iota}_s) \widehat{t})
& =
\sqrt{n}\mathbb{P}_{n} \left(  \left( Y - \widehat{m} \right)
(\phi_{s} -\what{\iota}_s) \widehat{t} \right)
\nonumber \\
&=
\sqrt{n}\mathbb{P}_{n} (\varepsilon (\phi_{s} - \what{\iota}_s)  \widehat{t})
+
\sqrt{n}\mathbb{P}_{n} \left(  (m -  \widehat{m}) (\phi_{s} - \what{\iota}_s)  \widehat{t} \right)
\, .
\label{eq:align3}
\end{align}
A similar reasoning as in Part (a) allows replacing $\widehat{t}$ by $t$ and yields
\begin{align*}
 \sqrt{n}\mathbb{P}_{n} (\widehat{\varepsilon} (\phi_{s} - \what{\iota}_s) \widehat{t})
& = \sqrt{n}\mathbb{P}_{n} (\varepsilon (\phi_{s} - \what{\iota}_s)  {t})
+
\sqrt{n}\mathbb{P}_{n} \left(  (m -  \widehat{m}) (\phi_{s} -
\what{\iota}_s) {t} \right)
+o_p^{\mathcal{S}}(1)
\, ,
\end{align*}
based on the boundedness of $\varepsilon$, $\sqrt{n} \mathbb{P}_{n} |\widehat{t} - t| = o_{p}(1)$, and the boundedness in probability of
$\sup_s \| (\phi_s - \what{\iota}_s)  {t}\|_\infty$, $\sup_s \| (\phi_s - \what{\iota}_s)  \widehat{t}\|_\infty$, $\|(\what{m}-m)t\|_\infty$, and $\|(\what{m}-m)\what{t}\|_\infty$, see Lemma \ref{lem:littleo}.
%\footnote{More specifically, $\sup_s \| (\phi_s - \what{\iota}_s)  {t}\|_\infty \leq \sup_s\|\phi_s-\iota_s\|_\infty + \sup_s\|(\iota_s-\what{\iota}_s)t\|_\infty$. The first term on the right-hand side is bounded, as $\phi_s$ and $\iota_s$ are uniformly bounded, while the second term is $o_P(1)$ from Lemma \ref{lem:littleo}.}
From Lemma \ref{lem:littleo}, $\sup_s \|(\what{\iota}_s - \iota_s)t\|_\infty =
o_{p}(n^{-1/4})$ and $\|(m-\what{m})t\|_\infty=o_P(n^{-1/4})$, hence
\[
\sqrt{n}\mathbb{P}_{n} \left(  (m -  \widehat{m}) (\phi_{s} -
\what{\iota}_s) {t} \right)
 =
\sqrt{n}\mathbb{P}_{n} \left(  (m -  \widehat{m}) (\phi_{s} -
{\iota}_s) {t} \right)
+o_p^{\mathcal{S}}(1)
=
o_p^{\mathcal{S}}(1)
\, ,
\]
where the last equality is established in Part (a).
Now,
\[
\mathbb{P}_n \varepsilon (\phi_s-\widehat{\iota}_s) t =
\mathbb{P}_n \varepsilon (\phi_s - \iota_s) t
- \mathbb{P}_n \varepsilon (\widehat{\iota}_s - \iota_s) t
\, ,
\]
and
$
\sup_s  \left|\sqrt{n}\mathbb{P}_{n} (\varepsilon (\phi_{s} - \iota_s) (t - 1)
) \right|
\leq C \sqrt{n}\mathbb{P}_{n}  |t -1|
= o_P(1)$
by Lemma \ref{lem:littleo}-(i).
Last,
$
\sqrt{n}\mathbb{P}_n \varepsilon (\widehat{\iota}_s - \iota_s) t = o_{p}^\mathcal{S}(1)
$
by Lemma \ref{lem: ASE}-(i).
Gathering results, $ \sqrt{n}\mathbb{P}_{n} (\widehat{\varepsilon} (\phi_{s} - \what{\iota}_s) \widehat{t})
 =  \sqrt{n}\mathbb{P}_{n} ({\varepsilon} (\phi_{s} - {\iota}_s) ) +
  o_p^{\mathcal{S}}(1) $.

\subsection{Proof of Proposition  \ref{prop: The Bootstrap Test}}
We here consider statements relative to $P^\xi\otimes P$, the joint probability measure of both the bootstrap weights and the sample data.

(a) We  study the bootstrap version of the bias corrected empirical process
\begin{align*}
 \sqrt{n}\mathbb{P}_{n} (\widetilde{\varepsilon}^{*}\phi_{s} \widehat{t})
& =
\sqrt{n}\mathbb{P}_{n} \left(  \left( Y^{*} - \widehat{m}^{*} + \widehat{B}^{*}
\right) \phi_{s} \widehat{t} \right)
\nonumber \\
&=
\sqrt{n}\mathbb{P}_{n} (\xi \varepsilon \phi_{s} \widehat{t})
+
\sqrt{n}\mathbb{P}_{n} \left( \xi (m -  \widehat{m}) \phi_{s} \widehat{t} \right)
+
\sqrt{n}\mathbb{P}_{n} \left( (\widehat{m} -  \widehat{m}^{*})
\phi_{s} \widehat{t} \right)
+
\sqrt{n}\mathbb{P}_{n} ( \widehat{B}^{*} \phi_{s} \widehat{t})
\, .
%\label{eq:align1b}
\end{align*}
From here we only stress the differences with the proof of Proposition 1.
%{\bf (Note  the results of Lemma \ref{lem:littleo} still hold if we replace,
%say, $(m -  \widehat{m})$ by $(\widehat{m} - \widehat{m}^{*})$.
%Maybe we should elaborate more?)}
%\textcolor{magenta}{(I would propose to write such results explicitly in Lemma \ref{lem:littleo})}
Proceeding as in the latter proof and using results in Lemma \ref{lem:littleo},
\begin{align*}
\sqrt{n}\mathbb{P}_{n} (\xi \varepsilon \phi_{s} \widehat{t})
 &= \sqrt{n}\mathbb{P}_{n} (\xi \varepsilon \phi_{s} {t})
+o_P^{\mathcal{S}}(1)
\, ,
 \\
\sqrt{n}\mathbb{P}_{n} \left( \xi (m -  \widehat{m}) \phi_{s} \widehat{t} \right)
&=
\sqrt{n}\mathbb{P}_{n} \left( \xi (m -  \widehat{m}) \phi_{s}
{t} \right)
+ o_P^{\mathcal{S}}(1)
\, ,
\\
\sqrt{n}\mathbb{P}_{n} \left( (\widehat{m} -  \widehat{m}^{*})
\phi_{s} \widehat{t} \right)
&=
\sqrt{n}\mathbb{P}_{n} \left( (\widehat{m} -  \widehat{m}^{*})
\phi_{s} {t} \right)
+ o_P^{\mathcal{S}}(1)
\, ,
\\
\sqrt{n}\mathbb{P}_{n} ( \widehat{B}^{*} \phi_{s} \widehat{t})
&  =
 \sqrt{n}\mathbb{P}_{n} ( \widehat{B}^{*} \phi_{s} {t})
+ o_P^{\mathcal{S}}(1)
\, .
\end{align*}

For the last term, use $\|\widehat{B}^{*} t \|_{\infty} =
o_{P}(n^{-1/4})$, Lemma \ref{lem:Belonging Conditions}-(ii)-(iv)-(v), and Lemma \ref{lem: ASE}-(i)
 to show that
\begin{equation}
\label{ASEb}
 \sqrt{n}\mathbb{P}_{n} (\widehat{B}^{*} t \phi_{s}) = \sqrt{n} P
 (\widehat{B}^{*} t \phi_{s}) + o_{P}^{\mathcal{S}}(1)
 \, .
\end{equation}
From the law of iterated expectations,
$ P (\widehat{B}^{*} t \phi_s) = P (\widehat{B}^{*}
t \iota_s)$.
Since
$\| [\widehat{B}^{*}-(\overline{\what{m}^{*}}-\overline{Y^{*}})/ f ] t \|_{\infty}
= o_P(n^{-1/2})$  from Lemma \ref{lem:littleo}-(viii),
\begin{align*}
\sqrt{n}  P ( \widehat{B}^{*} t \iota_{s})
& =
\sqrt{n} \int (\overline{\widehat{m}^{*}}(w)-\overline{Y^{*}}(w)) t(w)
\iota_s(w) \, dw
+ o_P^\mathcal{S}(1)
\\
& = - \sqrt{n}\mathbb{P}_{n} \left[ (Y^{*}-\widehat{m}^{*})\widehat{t}
\int K(u) (t \iota_{s})(W+uh)\, du \right]  + o_P^{\mathcal{S}}(1)
\, ,
\end{align*}
where the second equality follows from a change of variable.  Replace
$(Y^{*}-\what{m}^{*})\what{t}$ by $\xi\varepsilon {t} + \xi\varepsilon(\what{t}-t)+\xi(m-\what{m})\what{t}+(\what{m}-\what{m}^*)\what{t}$  and proceed as in Proposition 1 to obtain
\begin{align*}
\sqrt{n} P (\widehat{B}^{*} t \phi_{s})
& =
- \sqrt{n}\mathbb{P}_{n} (\xi \varepsilon \iota_st)
- \sqrt{n}\mathbb{P}_{n} ( (\widehat{m}-\what{m}^*) {t} \iota_s)
- \sqrt{n}\mathbb{P}_n\xi(m-\what{m})t\iota_s
+o_P^{\mathcal{S}}(1)
\, .
\end{align*}
Gathering results,
\begin{align*}
\sqrt{n}\mathbb{P}_{n} (\widetilde{\varepsilon}^{*}\phi_{s} \widehat{t})
 =&
\sqrt{n}\mathbb{P}_{n} (\xi \varepsilon (\phi_{s} - \iota_s) {t}) +
\sqrt{n}\mathbb{P}_{n} \left(\xi  (m -  \widehat{m}) (\phi_{s} - \iota_s)
{t}  \right)
\\
& \qquad
+ \sqrt{n}\mathbb{P}_n((\what{m}-\what{m}^*)t(\phi_s-\iota_s))
+o_P^{\mathcal{S}}(1)
\, .
\end{align*}
A reasoning similar to Proposition 1 then yields $
\sqrt{n}\mathbb{P}_{n} (\xi \varepsilon (\phi_{s} - \iota_s) )t =
\sqrt{n}\mathbb{P}_{n} (\xi \varepsilon (\phi_{s} - \iota_s) )
+o_P^{\mathcal{S}}(1)
$, $\sqrt{n}\mathbb{P}_{n} \left(\xi  (m -  \widehat{m}) (\phi_{s} - \iota_s)
{t}  \right) = o_P^{\mathcal{S}}(1)$, and $\sqrt{n}\mathbb{P}_n((\what{m}-\what{m}^*)t(\phi_s-\iota_s))=o_P^\mathcal{S}(1)$.
Hence
$$\sqrt{n}\mathbb{P}_n \widetilde{\varepsilon}^*\what{t}\phi_s=\sqrt{n}\mathbb{P}_n\xi\varepsilon(\phi_s-\iota_s)+o_P^\mathcal{S}(1).$$
%Since $o_\mathbb{P}^\mathcal{S}(1)=o_{P^*}^\mathcal{S}(1)$, the proof of the result is complete. \textcolor{magenta}{I am not %sure if we want to include this latter line: by including it we are more precise, but maybe it is a straightforward equality.}

For the locally robust empirical process,
\begin{align*}
 \sqrt{n}\mathbb{P}_{n} (\widehat{\varepsilon}^{*} (\phi_{s} - \what{\iota}_s) \widehat{t})
& =
\sqrt{n}\mathbb{P}_{n} (\xi \varepsilon (\phi_{s} - \what{\iota}_s)
\widehat{t})
\\ &
+
\sqrt{n}\mathbb{P}_{n} \left(  \xi (m -  \widehat{m}) (\phi_{s} -
\what{\iota}_s)  \widehat{t} \right)
+
\sqrt{n}\mathbb{P}_{n} \left(  (\widehat{m} - \widehat{m}^{*}) (\phi_{s} - \what{\iota}_s)  \widehat{t} \right)
\, .
\end{align*}
The proof proceeds along similar lines to show that the first term equals
$\sqrt{n}\mathbb{P}_{n} (\xi \varepsilon (\phi_{s} - {\iota}_s) ) +
o_P^{\mathcal{S}}(1)$ and the other terms
are both $o_P^{\mathcal{S}}(1)$.

%\textcolor{magenta}{I have revised the reminder of the proof. I hope it should be clear}

{(b)} Since the
class $\{(y,w,x)\mapsto(y-m(w))(\phi_s(x)-\iota_s(w)):s\in
\mathcal{S}\}$ is Donsker, $\sqrt{n}\mathbb{P}_n\varepsilon(\phi_s-\iota_s)$ converges weakly to a tight zero-mean Gaussian process $\mathbb{G}_s$ under $H_0$, see the main text. By the continuity of the Cramer-Von Mises functional and Proposition 1,  $S_n^{BC}$ and $S_n^{LR}$ weakly converge to
$\int|\mathbb{G}_s|^2\mu(ds)$. From \cite{bentkus_asymptotic_1993}, this distribution is continuous, so pointwise convergence implies  uniform convergence.

Using \citet[Theorem 2.9.6]{van_der_vaart_weak_1996}, we get  weak convergence of $\sqrt{n}\mathbb{P}_n \xi \varepsilon(\phi_s-\iota_s)$ in probability conditionally upon the initial sample to a tight zero-mean Gaussian process $\mathbb{G}'_s$ with the same covariance function as $\mathbb{G}_s$.
From (a), the bootstrap statistics (\ref{eq:bootstat})
weakly converge to $\int|\mathbb{G}'_s|^2\mu(ds)$ in probability conditionally upon the initial sample. The desired result then follows.

{(c)} From (a) and a Glivenko-Cantelli property of the class $\{(y,w,x)\mapsto(y-m(w))(\phi_s(x)-\iota_s(w)):s\in
\mathcal{S}\}$,
$\mathbb{P}_n\widetilde{\varepsilon}\what{t}\phi_s=
P \varepsilon (\phi_s - \iota_s) +o_P^\mathcal{S}(1)$ and similarly for
$\mathbb{P}_n\widehat{\varepsilon}(\phi_s-\what{\iota})\what{t}$.
Hence,  $S_n^{LR}/n$ and $S_n^{BC}/n \cvp \int|\E\varepsilon (\phi_s - \iota_s) |^2d\mu(s)$.
From \citet[Theorem 2.2]{bierens_econometric_2017}, under $H_1$, $\E\varepsilon (\phi_s - \iota_s) \neq 0 $ for almost all $s$ and  $\int \left|\E\varepsilon
(\phi_s - \iota_s) \right|^2 \, d\mu(s)>0$.
From (b), it holds that the bootstrap statistics are bounded in probability, and the result follows.

\subsection{Auxiliary Lemmas\label{sec:Auxiliary-Lemmas}}

\begin{lem}
\label{lem:littleo}
Under Assumptions \ref{Ass0}-\ref{Assumption: 3bis},
\begin{enumerate}[label=(\roman*)]
\item $\mathbb{P}_{n}|t-1|=o_{p}(n^{-1/2})$ and
$\mathbb{P}_{n}|\widehat{t}-t|=o_{p}(n^{-1/2})$,
\item $\|(\widehat{m}-m) t \|_{\infty} = o_p(n^{-1/4})$ and
 $\|(\widehat{m}-m) \what{t}\|_{\infty} = o_p(n^{-1/4})$,
\item $\|\widehat{B} t \|_{\infty} = o_p(n^{-1/4})$ and
 $\|\widehat{B} \what{t}\|_{\infty} = o_p(n^{-1/4})$,
\item
$\| [\widehat{B}-(\overline{\what{m}}-\overline{Y})/ f ] t \|_{\infty} = o_p(n^{-1/2})$,
\item $ \sup_s \|(\widehat{\iota}_s-\iota_s) t \|_{\infty} = o_p(n^{-1/4})$ and
 $ \sup_s \|(\widehat{\iota}_s-\iota_s) \what{t}\|_{\infty} = o_p(n^{-1/4})$,
 \item $\|(\what{m}^*-m)t\|_\infty=o_P(n^{-1/4})$ and $\|(\what{m}^*-m)\what{t}\|_\infty=o_P(n^{-1/4})$,
 \item $\|\what{B}^*t\|_\infty=o_P(n^{-1/4})$ and $\|\what{B}^*\what{t}\|_\infty=o_P(n^{-1/4})$,
 \item $\|[\what{B}^*-(\overline{\what{m}^*}-\overline{Y^*})/f]t\|_\infty=o_P(n^{-1/2})$.
\end{enumerate}
\end{lem}

\begin{proof}
$(i)$ Notice first that $|t-1|=\1(f(\cdot)<\tau_n)$ and
\begin{equation*}
    \E\mathbb{P}_n\1(f(W)<\tau_n)\leq \text{Pr}(f(W)\leq3\tau_n /2)=p_n=o(n^{-1/2})
\end{equation*}
where the last equality follows from Assumption \ref{Assumption: 3bis}(i). So, by Markov's inequality, $\mathbb{P}_n|t-1|=o_P(n^{-1/2})$.\\
To prove the second part of $(i)$, define the event $\mathcal{A}_C:=\{\|\widehat{f}-f\|_\infty\leq C d_n \}$. From Lemma \ref{Lemma : Uniform convergence of kernel sums} $\|\what{f}-\E\what{f}\|_\infty=O_P(\sqrt{(\log n)/(nh^p)})$. Standard bias manipulations ensure that $\|\E\what{f}-f\|_\infty=O(h^r)$. So, $\|\what{f}-f\|_\infty=O(d_n)$ and by choosing $C$ large enough $\text{Pr}(\mathcal{A}_C)$ can be made arbitrarily close to 1 for each large $n$. Over such event, since $d_n\tau_n^{-1}=o(1)$ (see Assumption \ref{Assumption: 3}), for each $n$ large enough $1-\tau_n^{-1}(\what{f}(w)-f(w))\leq 3/2$ for all $w\in\mathcal{W}$. Thus,
$$\1\Big( \what{f}(w)\geq \tau_n \Big)=\1\Big( f(w)\geq \tau_n\Big[ 1-\frac{\what{f}(w)-f(w)}{\tau_n} \Big] \Big)\geq \1\Big( f(w)\geq 3\tau_n/2 \Big)$$
so that $f(w)\geq(3/2)\tau_{n}$ implies $t(w)=\widehat{t}(w)=1$.
Accordingly, over $\mathcal{A}_{C}$ for large enough $n$
\[
 |\widehat{t}(w)-t(w)|\leq\1\Big( f(w)<3\tau_{n}/2\Big)\text{ for all }w\in\mathcal{W}\,
\]
and
\[
  \mathbb{P}_n|\what{t}-t|\leq \mathbb{P}_{n} \1\Big( f(W)<3\tau_{n}/2\Big)=O_P(p_n)=o_P(n^{-1/2})\, ,
\]
where the last two equalities follow from Markov's inequality and Assumption \ref{Assumption: 3bis}(i). Conclude by recalling that for $C$ large enough $\text{Pr}(\mathcal{A}_C)$ can be made arbitrarily close to 1 for each large $n$.

$(ii)$ Consider the event $\mathcal{A}_C$ previously defined and fix $\delta \in (0,1]$. Since $d_n\tau_n^{-1}=o(1)$, over such event for each large $n$ we have $1+2[f(w)-\what{f}(w)]/(\delta \tau_n)\leq 2 $ for all $w\in\mathcal{W}$. So,
\begin{equation*}
\1\Big( \what{f}(w)\geq \delta \tau_n/2 \Big)=\1\Big( f(w)\geq \frac{\delta \tau_n}{2} \Big[ 1+2\frac{f(w)-\what{f}(w)}{\delta \tau_n} \Big] \Big)\geq \1\Big( f(w)\geq \delta \tau_n \Big).
\end{equation*}

As already seen, by choosing $C$ large enough $\text{Pr}(\mathcal{A}_C)$ can be made arbitrarily close to 1 for each large $n$. So, we obtain that
\begin{equation}\label{eq: trimming inequality 1}
\text{ for each $\delta\in(0,1]$ wpa1 : }\1\Big( f(w)\geq \delta \tau_n \Big)\leq \1\Big( \what{f}(w) \geq \delta \tau_n/2 \Big)\text{ for all }w\in\mathcal{W},
\end{equation}
where wpa1 stands for "with probability approaching one". Switching the roles of $\what{f}$ and $f$, by a similar argument we obtain that
\begin{equation}\label{eq: trimming inequality 2}
    \text{ for each $\delta\in(0,1]$ wpa1 : } \1\Big( \what{f}(w)\geq \delta \tau_n \Big)\leq \1\Big( f(w) \geq \delta \tau_n/2 \Big)\text{ for all }w\in\mathcal{W}.
\end{equation}
Now, for any fixed $\delta\in(0,1]$ (\ref{eq: trimming inequality 1}) implies that with probability approaching one
\begin{align*}
(\what{m}-m)\text{ }\1(f(\cdot)\geq \delta \tau_n)=&\frac{\overline{Y}-m\what{f}}{f}\text{ } \1(f(\cdot)\geq \delta \tau_n) \\
&+\frac{\overline{Y}-m\what{f}}{f}\text{ }  \frac{f-\what{f}}{\what{f}} \text{ } \1(f(\cdot)\geq \delta \tau_n) \text{ } \1(\what{f}(\cdot)\geq \delta \tau_n/2)\, .
\end{align*}

To bound the RHS, from Lemma \ref{Lemma : Uniform convergence of kernel sums} $\|\overline{Y}-\E\overline{Y}\|_\infty=O_P(\sqrt{(\log n)/(nh^p)})$, while standard bias computations yield $\|\E\overline{Y}-mf\|_\infty=O(h^r)$. Hence, $\|\overline{Y}-mf\|_\infty=O_P(d_n)$. Similarly, $\|\what{f}-f\|_\infty=O_P(d_n)$ . Using these rates and the above display gives
\begin{equation}\label{eq: rate for mhat-m}
    \|(\what{m}-m)\text{ }\1(f(\cdot)\geq\delta\tau_n)\|_\infty=O_P\left( \frac{d_n}{\tau_n}+\frac{d_n^2}{\tau_n^2} \right)=o_P(n^{-1/4})
\end{equation}
where the last equality follows from Assumption \ref{Assumption: 3}. Since $\delta\in(0,1]$, the LHS of the above display is an upper bound for $\|(\what{m}-m)t\|_\infty$. So $\|(\what{m}-m)t\|_\infty=o_P(n^{-1/4})$. To obtain the rate with $\what{t}$, an application of (\ref{eq: trimming inequality 2}) with $\delta=1$ gives that  with probability approaching one $\|(\what{m}-m)\what{t}\|_\infty\leq \|(\what{m}-m)\text{ }\1(f(\cdot)\geq \tau_n/2)\|_\infty$. Applying (\ref{eq: rate for mhat-m}) with $\delta=1/2$ gives the $n^{-1/4}$ rate for the RHS of the latter inequality. 

$(iii)$ Since $\widehat{B}=\overline{\what{m}}/\what{f} - \what{m}$, in view of $(ii)$ it suffices to obtain a rate for $\overline{\what{m}}/\what{f}$.
To this end, notice that
\begin{equation} \label{eq: decomposition of mhatoverline}
\overline{\what{m}}(w)=\overline{m}(w)+\overline{m(\what{t}-1)}(w)+\overline{(\what{m}-m)\what{t}}(w) 
\, .
\end{equation}

Combining Lemma \ref{Lemma : Uniform convergence of kernel sums} with standard bias computations gives $\|\overline{m}-mf\|_\infty=O_P(d_n)$. For the second term, the boundedness of $K$ and $m$ implies $|\overline{m(\widehat{t}-1)}(w)|\leq C h^{-p} \mathbb{P}_n|\widehat{t}-1|$. Using the same arguments as in the proof of $(i)$, $\mathbb{P}_n|\widehat{t}-1|\leq \mathbb{P}_n|\widehat{t}-t| + \mathbb{P}_n|t-1| =O_P(p_n)$. So, $\|\overline{m(\what{t}-1)}\|_\infty=O_P(p_nh^{-p})$. For the third term on the RHS of (\ref{eq: decomposition of mhatoverline})
$$\overline{(\what{m}-m)\what{t}}(w)\leq \|(\what{m}-m)\what{t}\|_\infty h^{-p} \mathbb{P}_n\Big|K\Big(\frac{W-w}{h}\Big)\Big|=O_P\Big(\frac{d_n}{\tau_n}\Big)\cdot h^{-p}\mathbb{P}_n\Big|K\Big(\frac{W-w}{h}\Big)\Big|$$
where in the last equality we have used $\left\| (\what{m}-m)\what{t}\right\|_\infty=O_P(d_n/\tau_n) $ from the proof of $(ii)$. An application of Lemma \ref{Lemma : Uniform convergence of kernel sums} and standard bias computations yield that the   last term on the RHS is $O_P(1)$ uniformly in $w\in\mathcal{W}$.  Gathering results,
\begin{equation}\label{eq: mhatoverline}
    \|\overline{\what{m}}-mf\|_\infty=O_P\left( d_n+p_nh^{-p}+\frac{d_n}{\tau_n}  \right) \, .
\end{equation}
Combining the above display with arguments similar to the proof of $(ii)$ gives
\begin{equation}\label{eq: mhatoverline/fhat}
    \left\| \left( \frac{\overline{\what{m}}}{\what{f}} - m \right) \what{t}  \right\|_\infty =O_P\left( \frac{d_n}{\tau_n^2} + \frac{p_n}{h^p\tau_n}\right)\text{ and }\left\| \left( \frac{\overline{\what{m}}}{\what{f}} - m \right) {t}  \right\|_\infty =O_P\left( \frac{d_n}{\tau_n^2} + \frac{p_n}{h^p\tau_n}\right) \, .
\end{equation}
Using Assumption \ref{Assumption: 3} and \ref{Assumption: 3bis}(i) gives the desired result.

$(iv)$ Since $\what{B}=\overline{\what{m}}/\what{f}-\overline{Y}/\what{f}$, Equation (\ref{eq: trimming inequality 1}) implies that with probability approaching one
\begin{equation*}
\Big( \what{B}-\frac{\overline{\what{m}}-\overline{Y}}{f} \Big)\text{ }t=(\overline{\what{m}}-\overline{Y})\Big( \frac{f-\what{f}}{\what{f}f}\text{ } \Big)\text{ }t\text{ }\1(\what{f}(\cdot)\geq \tau_n/2)  \, .
\end{equation*}
Using (\ref{eq: mhatoverline}), $\|\overline{Y}-mf\|_\infty=O_P(d_n)$, and $\|\what{f}-f\|_\infty=O_P(d_n)$ (see the proof of $(ii)$), the RHS of the above display is
$$O_P\Big( \Big( \frac{d_n}{\tau_n^2} + \frac{p_n}{h^p\tau_n} \Big)\text{ }\frac{d_n}{\tau_n} \Big)=o_P(n^{-1/2})$$
uniformly in $w\in\mathcal{W}$, where the last equality follows from Assumptions \ref{Assumption: 3} and \ref{Assumption: 3bis}(i). 

$(v)$ The proof follows from arguments similar to the proof of $(ii)$, so it is omitted.

$(vi)$ First, notice that
\begin{equation}\label{eq: decomposition of mhatsar}
    \what{m}^*=\frac{\overline{Y^*}}{\what{f}}=\frac{\overline{\what{m}}}{\what{f}}+ \frac{\overline{\xi\what{\varepsilon}}}{\what{f}}.
\end{equation}
The uniform convergence rate of $\overline{\what{m}}/\what{f}$ has already been obtained in (\ref{eq: mhatoverline/fhat}), so it suffices to show that the second addendum is negligible with a suitable rate. Now, $\overline{\xi\what{\varepsilon}}=\overline{\xi\varepsilon}+\overline{\xi\varepsilon(\what{t}-1)}+\overline{\xi(m-\what{m})\what{t}}$.  Using this decomposition and proceeding as in the proof of  $(iii)$ gives
\begin{equation}\label{eq: rate for overlinexiepshat/fhat}
    \left\| \frac{\overline{\xi\what{\varepsilon}}}{\what{f}}\what{t} \right\|_\infty =O_P\left( \frac{d_n}{\tau_n^2} + \frac{p_n}{h^p\tau_n}\right)\text{ and } \left\| \frac{\overline{\xi\what{\varepsilon}}}{\what{f}}{t} \right\|_\infty=O_P\left( \frac{d_n}{\tau_n^2} + \frac{p_n}{h^p\tau_n}\right)
\end{equation}
with $d_n\tau_n^{-2}+p_nh^{-p}\tau_n^{-1}=o(n^{-1/4})$.

$(vii)$ Recall that
\begin{equation*}
    \what{B}^*=\frac{\overline{\what{m}^*}}{\what{f}}-\what{m}^*\text{ . }
\end{equation*}
In view of $(iv)$, it suffices to obtain a suitable convergence rate for the first addendum. To this end, from (\ref{eq: mhatoverline/fhat}), (\ref{eq: decomposition of mhatsar}), and (\ref{eq: rate for overlinexiepshat/fhat}) we have
\begin{equation}\label{eq: rate for mhatstar-m}
    \| (\what{m}^*-m)\what{t} \|_\infty =O_P\left( \frac{d_n}{\tau_n^2}+\frac{p_n}{h^p\tau_n} \right)\text{ .}
\end{equation}
Also, $\overline{\what{m}^*}=\overline{m}+\overline{m(\what{t}-1)}+\overline{(\what{m}^*-m)\what{t}}$. Using this decomposition, (\ref{eq: rate for mhatstar-m}), and proceeding similarly as in the proof of $(iii)$ gives
\begin{equation*}
    \left\|  \left( \frac{\what{m}^*}{\what{f}} - m \right)\what{t} \right\|_\infty  =O_P\left( \frac{d_n}{\tau_n^3} + \frac{p_n}{h^p\tau_n^2} \right)\text{ and } \left\|  \left( \frac{\what{m}^*}{\what{f}} - m \right){t} \right\|_\infty  =O_P\left( \frac{d_n}{\tau_n^3} + \frac{p_n}{h^p\tau_n^2} \right)\text{ . }
\end{equation*}
By Assumptions \ref{Assumption: 3} and \ref{Assumption: 3bis}(i) we obtain the desired result.

$(viii)$ The proof combines arguments already used previously. For completeness, we also provide it here. Using the definition of $\what{B}^*$ and (\ref{eq: trimming inequality 1}), with probability approaching one
\begin{equation}\label{eq: Bhatstar decomposition}
\Big( \what{B}^*-\frac{\overline{\what{m}^*}-\overline{Y^*}}{f} \Big)\text{ }t=(\overline{\what{m}^*}-\overline{Y^*})\Big( \frac{f-\what{f}}{f\what{f}} \Big)\text{ }t \text{ }\1(\what{f}(\cdot)\geq\tau_n/2)\text{ .}
\end{equation}
As noticed in the proof of $(vii)$,  $\overline{\what{m}^*}=\overline{m}+\overline{m(\what{t}-1)}+\overline{(\what{m}^*-m)\what{t}}$. By this decomposition, (\ref{eq: rate for mhatstar-m}), and reasoning as in the proof of $(iii)$ we get
\begin{equation}\label{eq: mhatstaroverline - mf}
\| \overline{\what{m}^*}-mf \|_\infty=O_P\Big( \frac{d_n}{\tau_n^2} + \frac{p_n}{h^p\tau_n} \Big).
\end{equation}
Combining $\overline{\xi\what{\varepsilon}}=\overline{\xi\varepsilon}+\overline{\xi\varepsilon(\what{t}-1)}+\overline{\xi(m-\what{m})\what{t}}$ with arguments used in the proof of $(iii)$ gives
$$ \| \overline{\xi\what{\varepsilon}} \|_\infty=O_P\left( \frac{d_n}{\tau_n}+\frac{p_n}{h^p} \right). $$
By the previous display, $\overline{Y^*}=\overline{\what{m}}+\overline{\xi\what{\varepsilon}}$, and (\ref{eq: mhatoverline}), we get
\begin{equation}\label{eq: Ystaroverline - mf}
\| \overline{Y^*}-mf \|_\infty=O_P\Big( \frac{d_n}{\tau_n} + \frac{p_n}{h^p} \Big).
\end{equation}
Finally, plugging $\|\what{f}-f\|_\infty=O_P(d_n)$,  (\ref{eq: mhatstaroverline - mf}), and (\ref{eq: Ystaroverline - mf}) into (\ref{eq: Bhatstar decomposition}) and then using Assumptions \ref{Assumption: 3} and \ref{Assumption: 3bis}(i) gives the desired result.
\end{proof}

%\textcolor{magenta}{I tried to write the following lemma in a simpler way}
\begin{lem}
\label{Lemma : Uniform convergence of kernel sums} Let $\{U_i\}_{i=1}^n$ be a sequence of i.i.d. random variables taking values in $\mathbb{R}^q$.
Assume that
$\{\varphi_{n,s}\text{ : }s\in\mathcal{S}\}$ is
a sequence of classes of real-valued functions defined on the support of $U_1$
such that for any $n\in\mathbb{N}$ : $\sup_{s\in\mathcal{S}}\|\varphi_{n,s}\|_{\infty}<L_{\varphi}$
and $ \|\varphi_{n,s_{1}}-\varphi_{n,s_{2}}\|_{\infty}\leq L_{\varphi}\|s_{1}-s_{2}\|$
for all $s_{1},s_{2}\in\mathcal{S}$. Then, for any compact set $\mathcal{A}\subset\mathbb{R}^{p}$
\[
 \sup_{(w,\theta)\in\mathcal{A}\times \mathcal{S}} \Big| h^{-p} (\mathbb{P}_{n}-P) \varphi_{n,s}(U) K\Big(\frac{W-w}{h}\Big) \Big|=O_{P}\Big(\sqrt{\frac{\log n}{nh^{p}}}\Big)\, .
\]
\end{lem}
\begin{proof} The result is a minor modification of the proof of Theorem 1.4 in \citet{li_nonparametric_2006}.
\end{proof}

%\textcolor{magenta}{I have removed the old Lemma2.2 (see AppendixPL): I got rid of Lemma2.2-(i), as the proof was two lines, and %I put lemma 2.2-(ii) in the last lemma of the present Appendix.}

%\bigskip
%{\bf The following lemma belongs to a supplementary appendix. Maybe make these
%a formal statement, and say we are proving it.}
%\textcolor{magenta}{I totally agree to put the proof of the following lemma in a supplementary appendix, as it is long. So, we just include its statement and then prove it in a supplementary material.}\\
%\vspace{0.5cm}

The following lemma provides the regularity features needed to apply stochastic equicontinuity results. Similar results can be found in  \citet{AndrewsHB} and \citet{andrews_nonparametric_1995}. The differences with respect to these works are the presence of the bias correction components $\overline{\what{m}}/\what{f}$ and the  different assumptions on the bandwidths. Recall that
$$\mathcal{W}_n:=\left\{ w\text{ : }f_W(w)\geq \tau_n/2\right\}\text{ . }$$
%\vspace{0.5cm}

\begin{lem}
\label{lem:Belonging Conditions} Under Assumptions \ref{Assumption: 1}-\ref{Assumption: 3bis} for $l =\lceil (p+1)/2 \rceil$,
\begin{enumerate}[label=(\roman*)]
    \item $\Pr \Big( \what{m}\in \mathcal{G}_l(\mathcal{W}_{n})\Big)\rightarrow1$,
    \item $\Pr \Big( \overline{\what{m}}/\what{f}\in \mathcal{G}_l(\mathcal{W}_{n})\Big)\rightarrow1$,
    \item $\Pr \Big( \what{\iota}_s\in \mathcal{G}_l(\mathcal{W}_{n})\text{ for all }s\in\mathcal{S}\Big)\rightarrow1$,
    \item $\Pr \Big( \overline{\xi \what{\varepsilon}}/\what{f}\in \mathcal{G}_l(\mathcal{W}_{n})\Big)\rightarrow1$,
    \item $\Pr \Big( \overline{\what{m}^*}/\what{f}\in \mathcal{G}_l(\mathcal{W}_{n})\Big)\rightarrow1$.
\end{enumerate}
\end{lem}
\begin{proof}
The proof of this result can be found in the Supplementary Material.
\end{proof}
%\vspace{0.5cm}

The following lemma provides a stochastic equicontinuity result. Similar results can be found in \citet{AndrewsHB} or \citet{andrews_nonparametric_1995}. Let us introduce some notation that will be used in the proof. For a generic space of functions $\mathcal{F}$ endowed with the $L_2(P)$ metric $||\cdot||_{2,P}$, we denote by $N(\epsilon,\mathcal{F},||\cdot||_{2,P})$ the $\epsilon$ covering number and with $N_{[\cdot]}(\epsilon,\mathcal{F},||\cdot||_{2,P})$ the $\epsilon$ bracketing number, see \citet{vaart_asymptotic_1998} and \citet{van_der_vaart_weak_1996}.

\begin{lem}
\label{lem: ASE}Let Assumptions \ref{Ass0}-\ref{Assumption: 3bis} hold and denote with $Supp(\zeta)$ the support of $\zeta:=(\xi,Y,W,X)$.
\begin{enumerate}[label=(\roman*)]
    \item If $\{\widehat{f}_{s}:s\in\mathcal{S}\}$ is a collection of stochastic real-valued functions defined on $\mathcal{W}$ such that $ \sup_{s\in\mathcal{S}}\| \widehat{f}_{s}t\|_\infty=o_P(1)$ and $\Pr(\widehat{f}_{s}\in \mathcal{G}_{l}(\mathcal{W}_{n}) \text{ for all }s\in\mathcal{S})\rightarrow1$ for $l =\lceil (p+1)/2 \rceil$, then
\[
 \sup_{s\in\mathcal{S}}\Big| \mathbb{G}_{n} t \widehat{f}_{s} \phi_{s}\Big|=o_{P}(1)\, .
\]
The same result also holds when $\phi_s$ is replaced by $\phi_s-\iota_s$, $\xi(\phi_s-\iota_s)$, or a fixed bounded function.

  \item  If $g_{n}:Supp(\zeta)\mapsto\mathbb{R}$ is a sequence of functions such that $\|g_n\|_{\infty}=O(1)$,
\[
 \mathbb{G}_{n} g_{n} \int K(u) \iota_s(\cdot+uh)\, du=\mathbb{G}_{n} g_{n} \iota_s+o_{P}^{\mathcal{S}}(1)\, .
\]
\end{enumerate}

\end{lem}

\begin{proof}
$(i)$ Fix an arbitrary $\epsilon>0$. The assumptions on $\{f_s:s\in\mathcal{S}\}$ ensure that  with probability approaching one
\begin{equation}\label{eq: Upperbound to introduce VdV bracketing maximal inequality}
  \sup_{s\in\mathcal{S}}\Big| \mathbb{G}_{n} t \what{f}_{s} \phi_{s}\Big|\leq \sup_{g\in \mathcal{G}_n^{\epsilon}} \left| \mathbb{G}_{n}g \right|
\end{equation}
with
$\mathcal{G}_{n}^{\epsilon}:=\{g\in\mathcal{G}_{n} : \|g\|_{2,P}<\epsilon\}$,
$\mathcal{G}_n = t\cdot \mathcal{G}_l(\mathcal{W}_n) \cdot \Psi$, and  $\Psi:=\left\{ \phi_s:s\in\mathcal{S}\right\}$.
By \citet[Lemma 19.34]{vaart_asymptotic_1998}, if
\begin{equation*}
    \log N_{[\cdot]}(\epsilon,\mathcal{G}_n,||\cdot||_{2,P})\leq C \epsilon^{-\gamma} \quad \text{ for all }\epsilon\in(0,1)\text{ and some fixed }\gamma\in(0,2)
    \, ,
\end{equation*}
then the expectation of the right-hand side of (\ref{eq: Upperbound to introduce VdV bracketing maximal inequality}) can be made arbitrarily small asymptotically by choosing $\epsilon$ small enough. So, Markov's inequality would deliver the desired result.

Since  $N_{[\cdot]}(2\epsilon,\mathcal{G}_n,||\cdot||_{2,P})\leq N(\epsilon,\mathcal{G}_n,||\cdot||_{\infty})$, we focus on the latter. Assumption \ref{Assumption: 3bis}(ii) and \citet[Theorem 2.7.1]{van_der_vaart_weak_1996} ensure that for each $n$ large enough
\[
  \log N\Big( \epsilon\, ,\mathcal{G}_{l}(\mathcal{W}_{n}) \, ,\|\cdot\|_{\infty,\mathcal{W}_{n}}\Big)\leq C\varepsilon^{-\upsilon},\quad \upsilon\in(0,2)\,
\]
where $||g||_{\infty,\mathcal{W}_n}:=\sup_{w\in\mathcal{W}_n}|g(w)|$ for any $g\in\mathcal{G}_l(\mathcal{W}_n)$.
Since $\phi_s$ is Lipschitz on the compact $\mathcal{S}$,
$N(\epsilon,\Psi,\|\cdot\|_\infty)\leq C\epsilon^{-q}$ by \citet[Theorem 9.15]{kosorok_introduction_2008}.
So, using the fact that $t$ is bounded,
\[
 N(\epsilon,\mathcal{G}_{n},\|\cdot\|_{\infty})\leq C \epsilon^{-q}\text{exp}(\epsilon^{-\upsilon})\text{ with \ensuremath{\upsilon\in(0,2)} }
 \, .
\]
The same reasoning applies if $\phi_s-\iota_s$ replaces $\phi_s$, and if $\xi(\phi_s-\iota_s)$ replaces $\phi_s$, since $N_{[\cdot]}( \epsilon ||\xi||_{2}, \xi \cdot \mathcal{G}_n, ||\cdot||_{2,P})\leq N_{[\cdot]} (\epsilon,\mathcal{G}_n,||\cdot||_{2,P}) $, where $||\xi||_2^2=\E|\xi|^2$.

$(ii)$ By a $r$th order Taylor expansion, $\int K(u) \iota_{s}(w+uh)\, du=\iota_{s}(w)+O(h^{r})$ uniformly in $w\in \mathcal{W}$.
Thus, the proof proceeds along similar lines as the proof of $(i)$.
%\\
%\vspace{0.5cm}

\end{proof}

\footnotesize

\paragraph{Acknowledgements}
We thank the Editor Michael Jansson and two referees for their comments that helped to improve the paper.
% We thank the Editor Michael Jansson and two referees and the associate editor for helpful comments that helped to improve the paper.

%\newpage
\bigskip
\bibliographystyle{ecta}
%\nocite{*}
\footnotesize
\renewcommand{\baselinestretch}{1}
\bibliography{biblio}
\normalsize

\newpage
\begin{figure}[ht]
\begin{center}
\includegraphics[width=7cm,height=7.5cm]{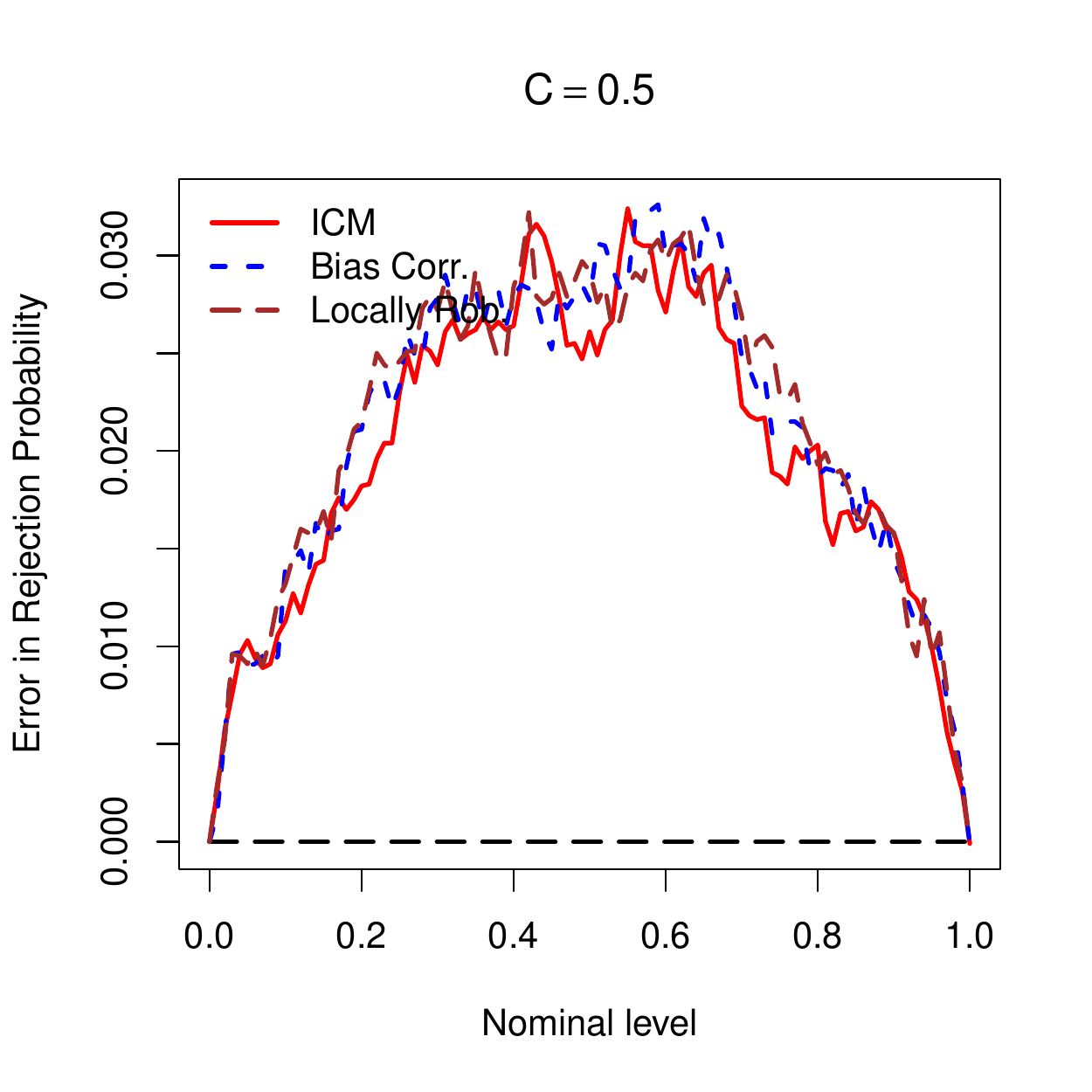}
\includegraphics[width=7cm,height=7.5cm]{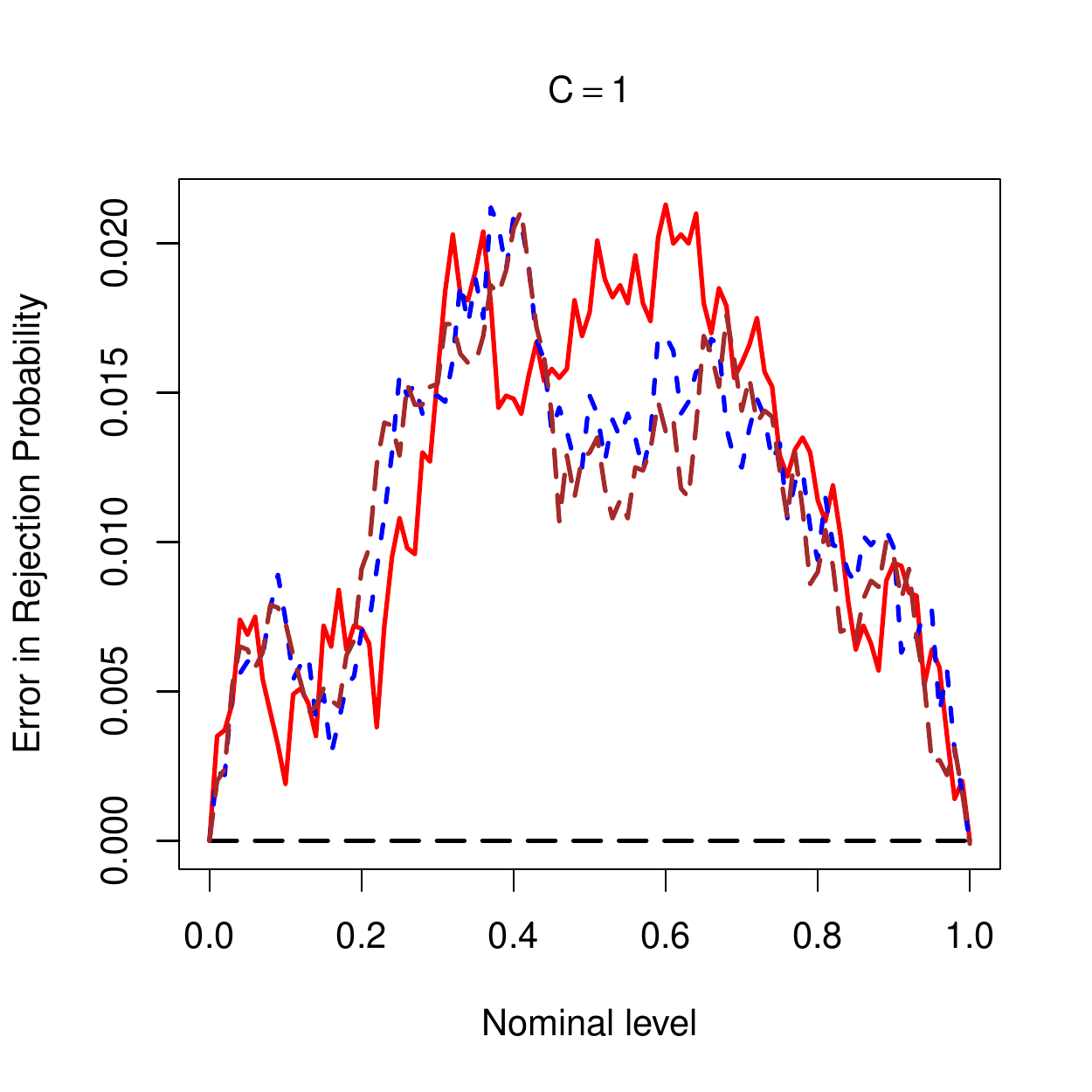}
\includegraphics[width=7cm,height=7.5cm]{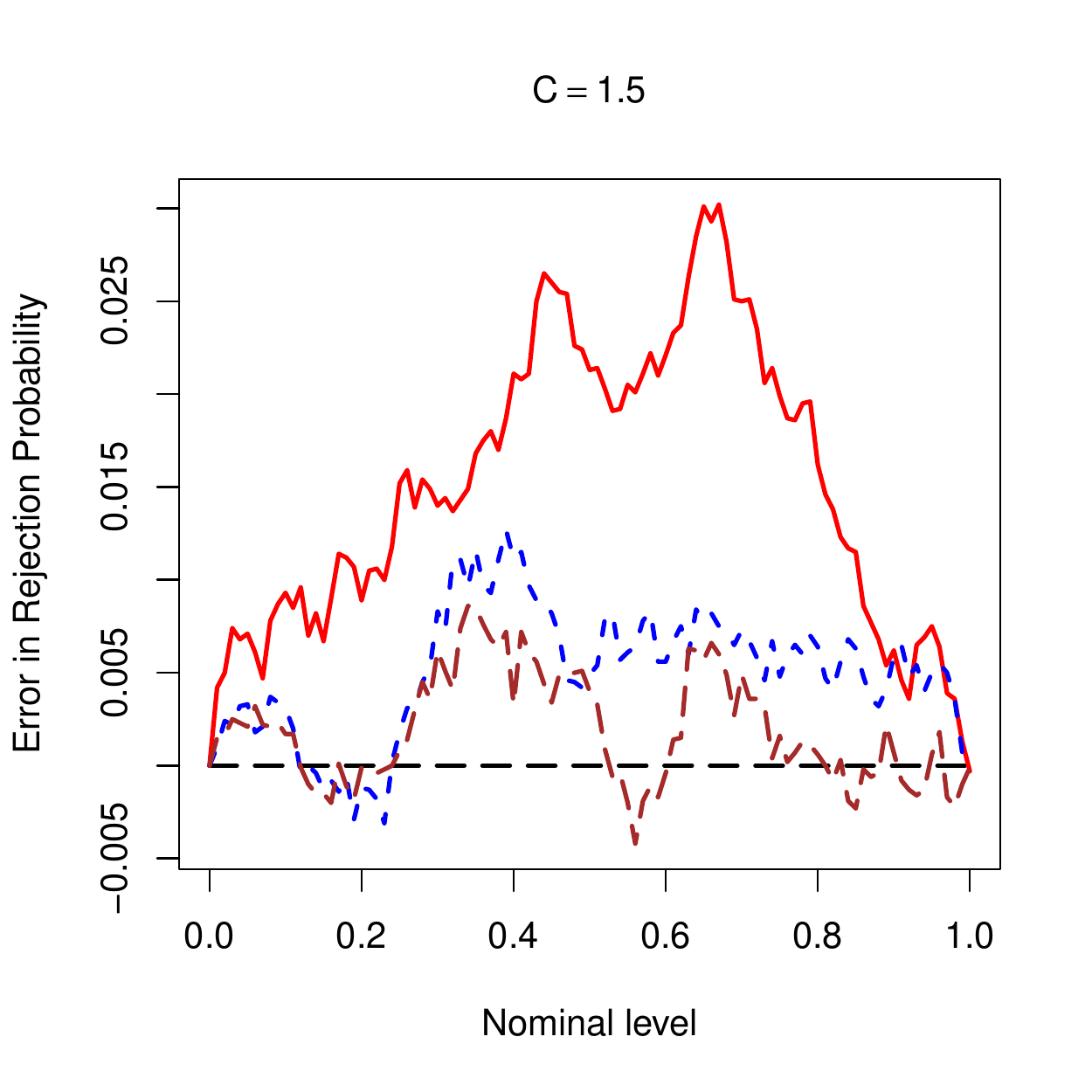}
\includegraphics[width=7cm,height=7.5cm]{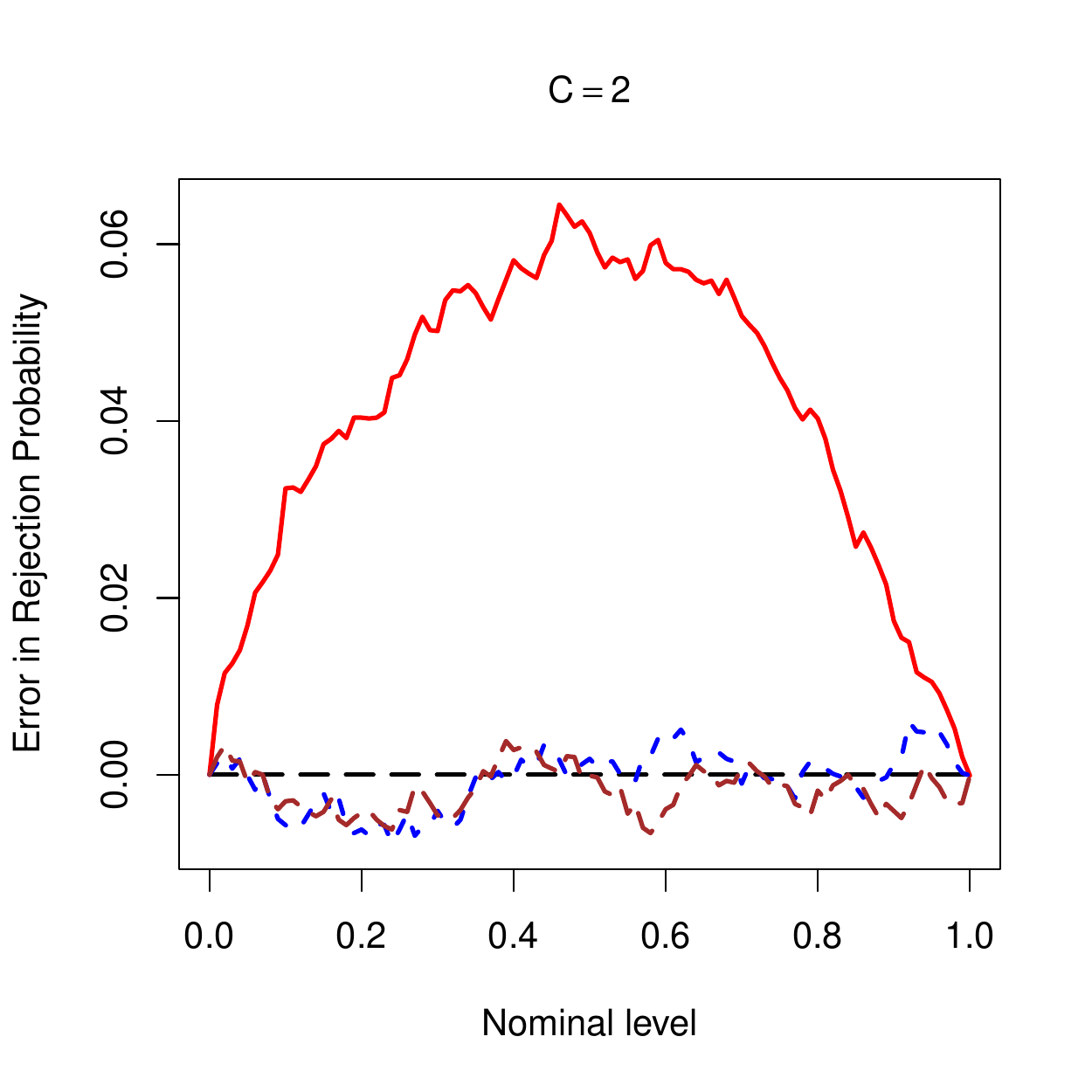}
\end{center}
\caption{\small Errors in rejection probabilities for $n=400$ using a
rule-of-thumb bandwidth $h=C\widehat{\sigma}_W n^{-1/5}$.}
\label{Fig o}
\end{figure}

\newpage
\begin{figure}[ht]
\begin{center}
\includegraphics[width=7cm,height=7.5cm]{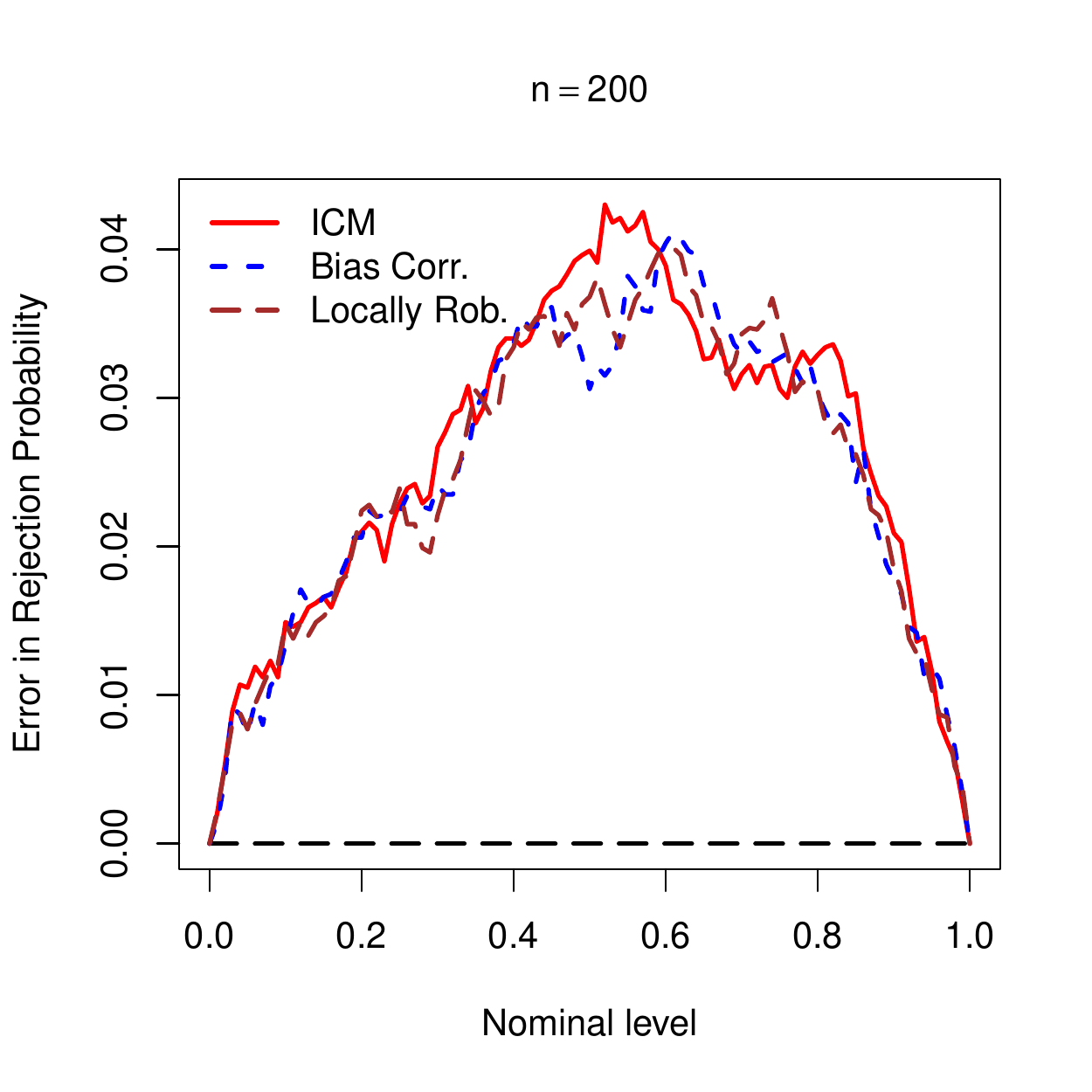}
\includegraphics[width=7cm,height=7.5cm]{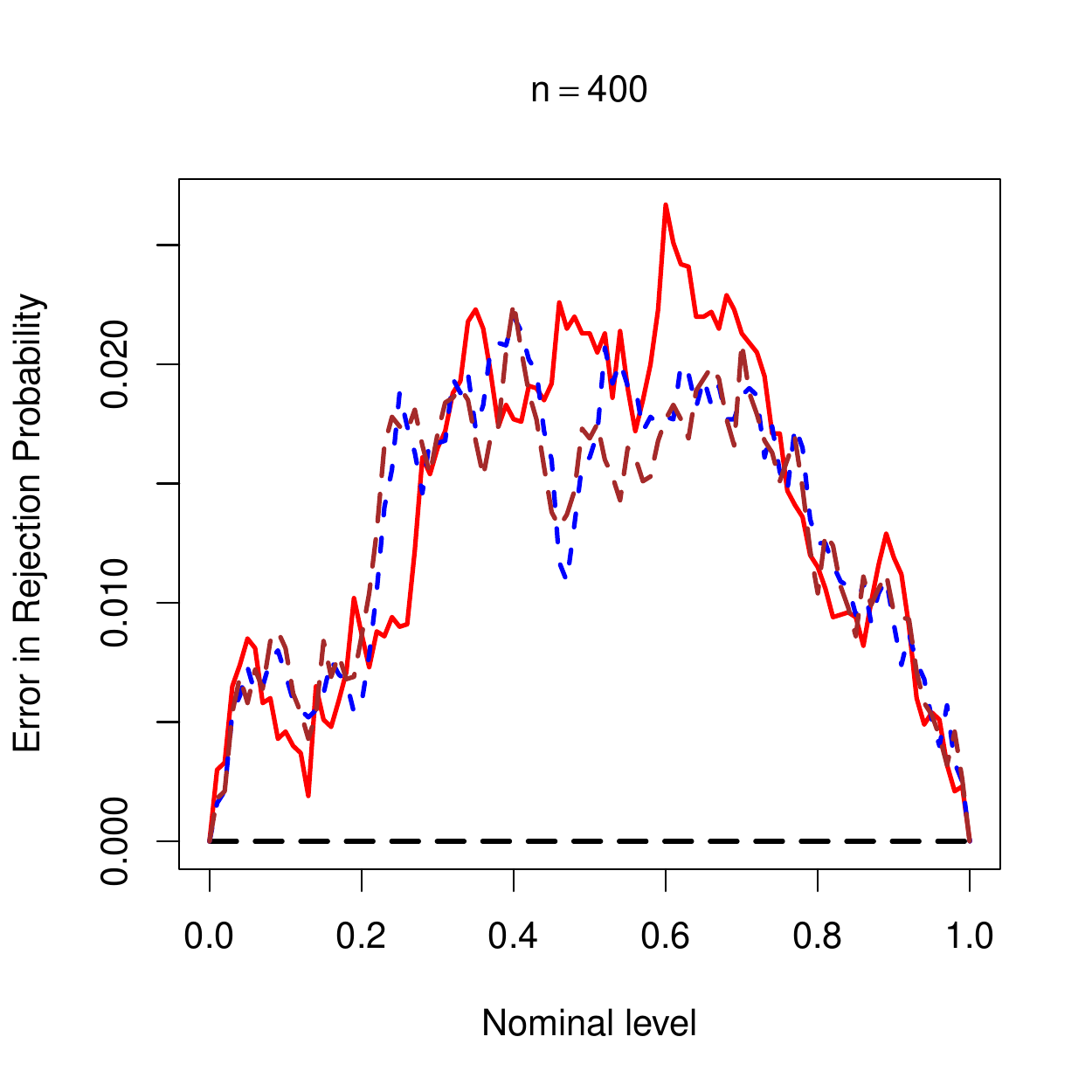}
\includegraphics[width=7cm,height=7.5cm]{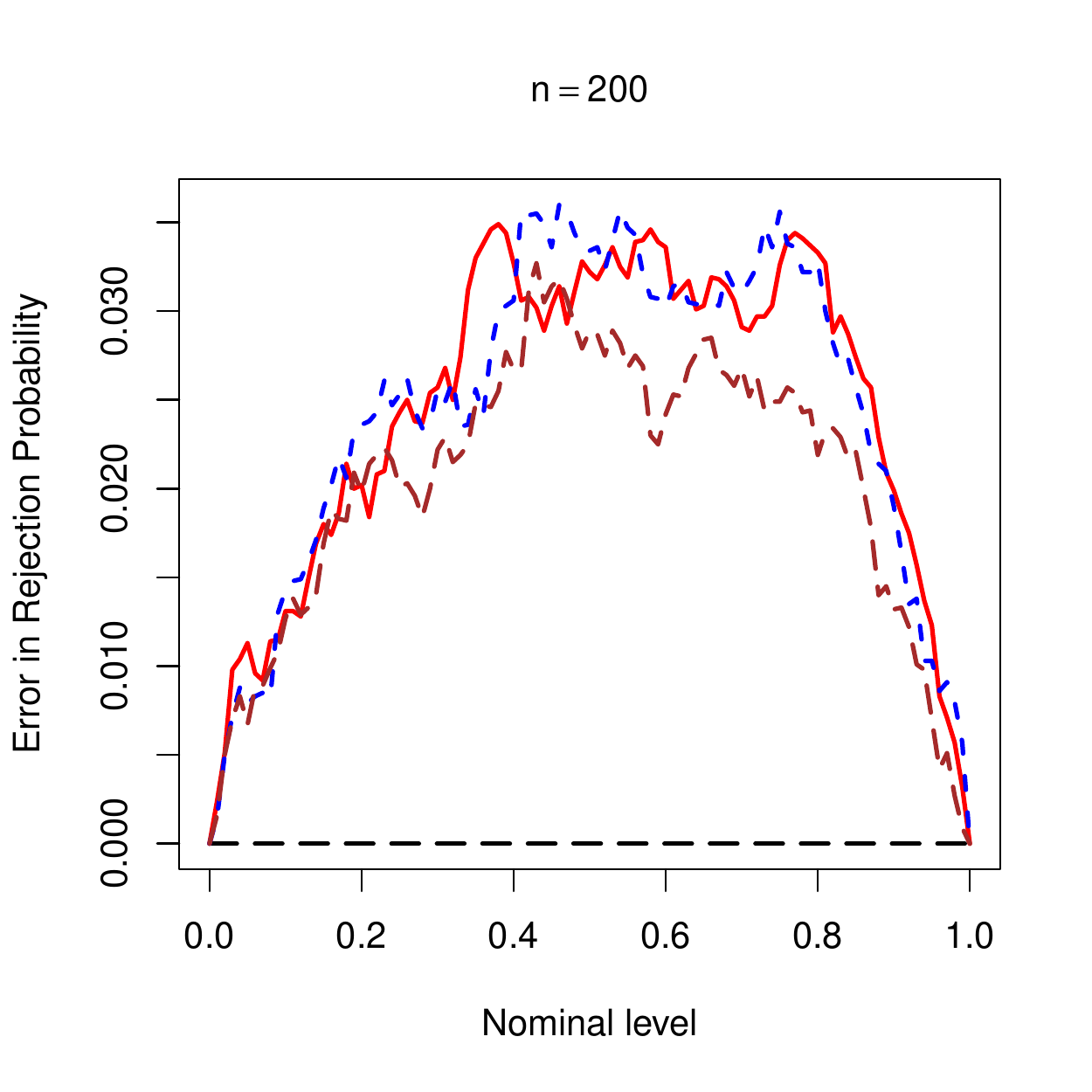}
\includegraphics[width=7cm,height=7.5cm]{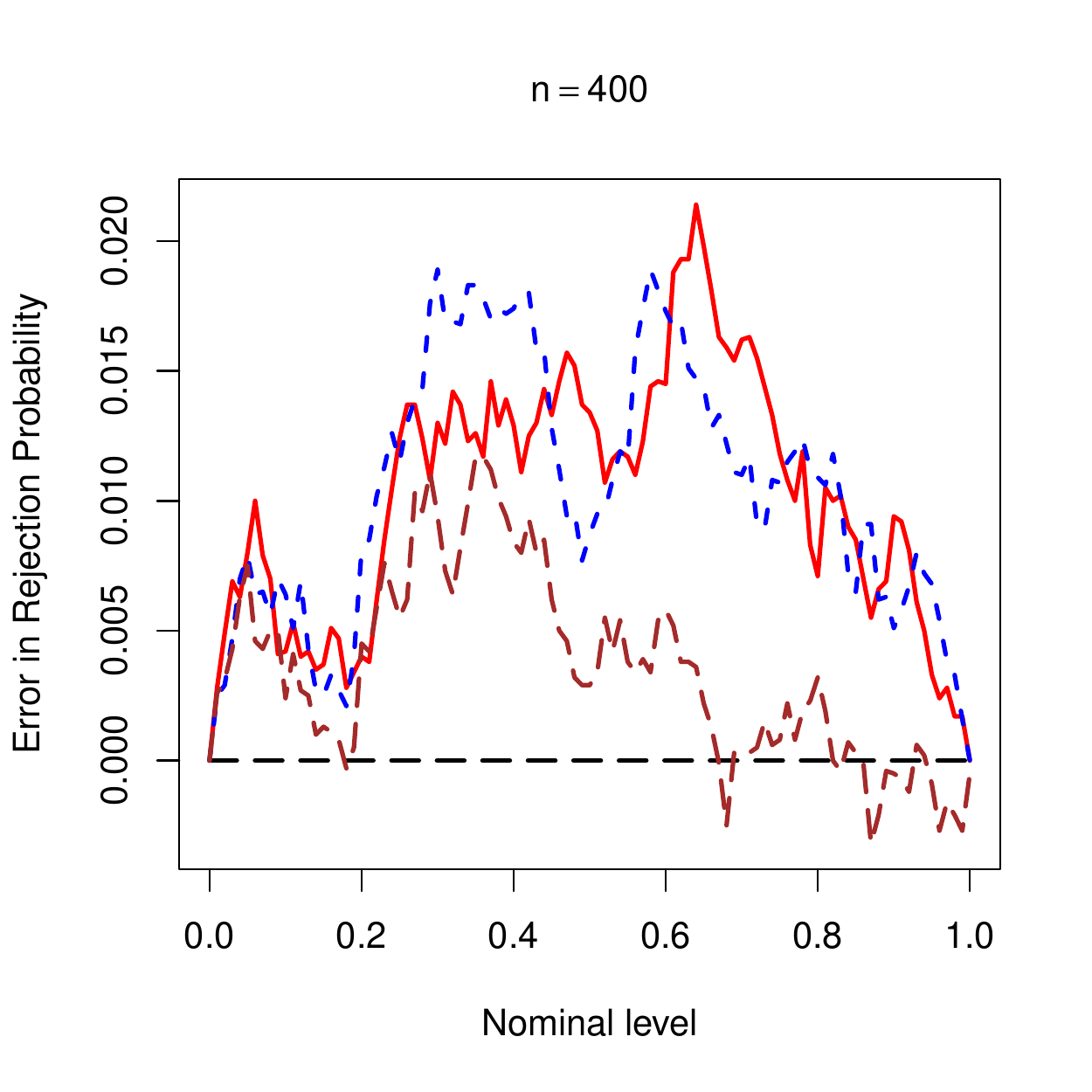}
\end{center}
\caption{\small Errors in rejection probabilities using a data-driven
bandwidth: results with trimming on first row, results with no
trimming on second row.}
\label{Fig i}
\end{figure}

\newpage
\begin{figure}[ht]
\begin{center}
\includegraphics[width=7cm,height=7.5cm]{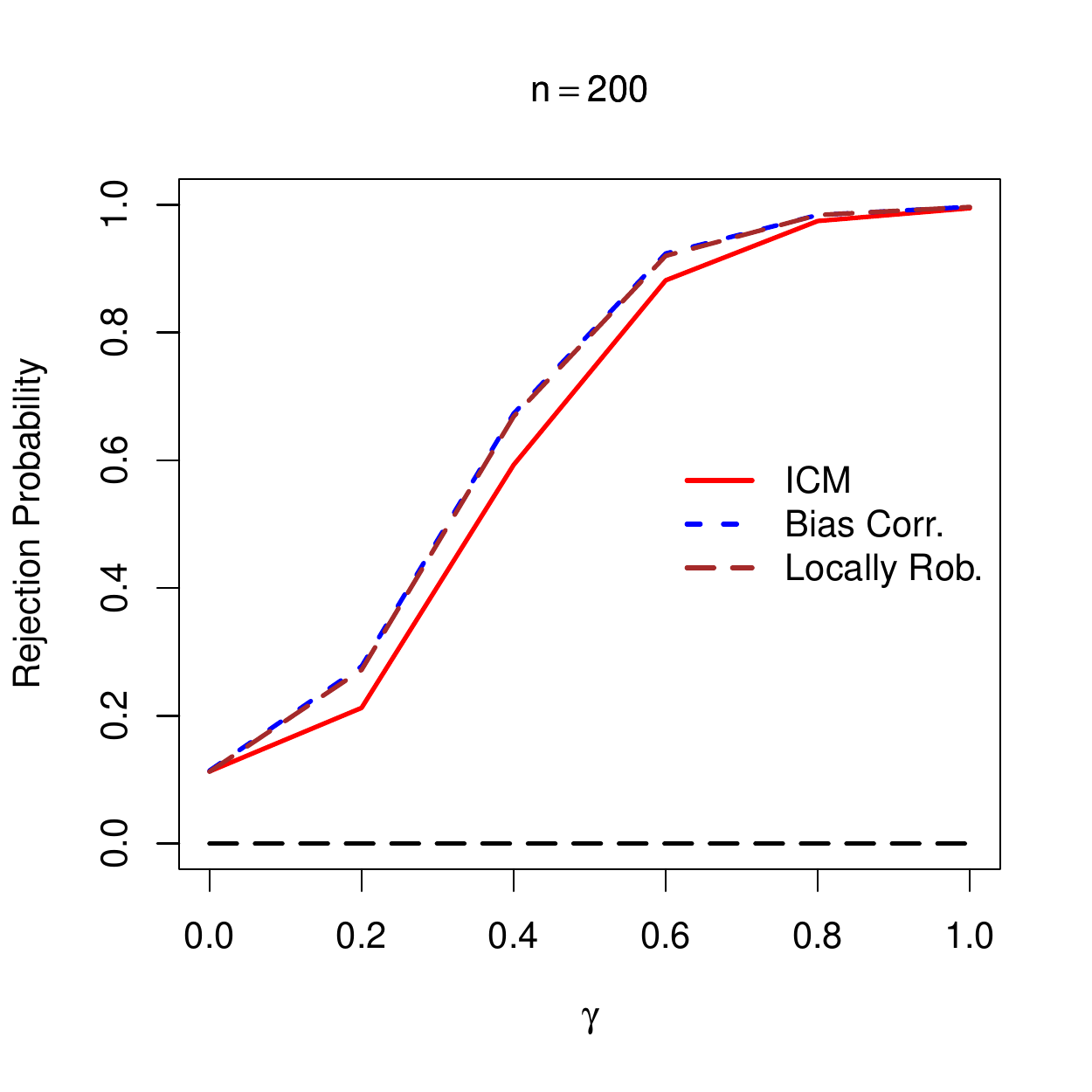}
\includegraphics[width=7cm,height=7.5cm]{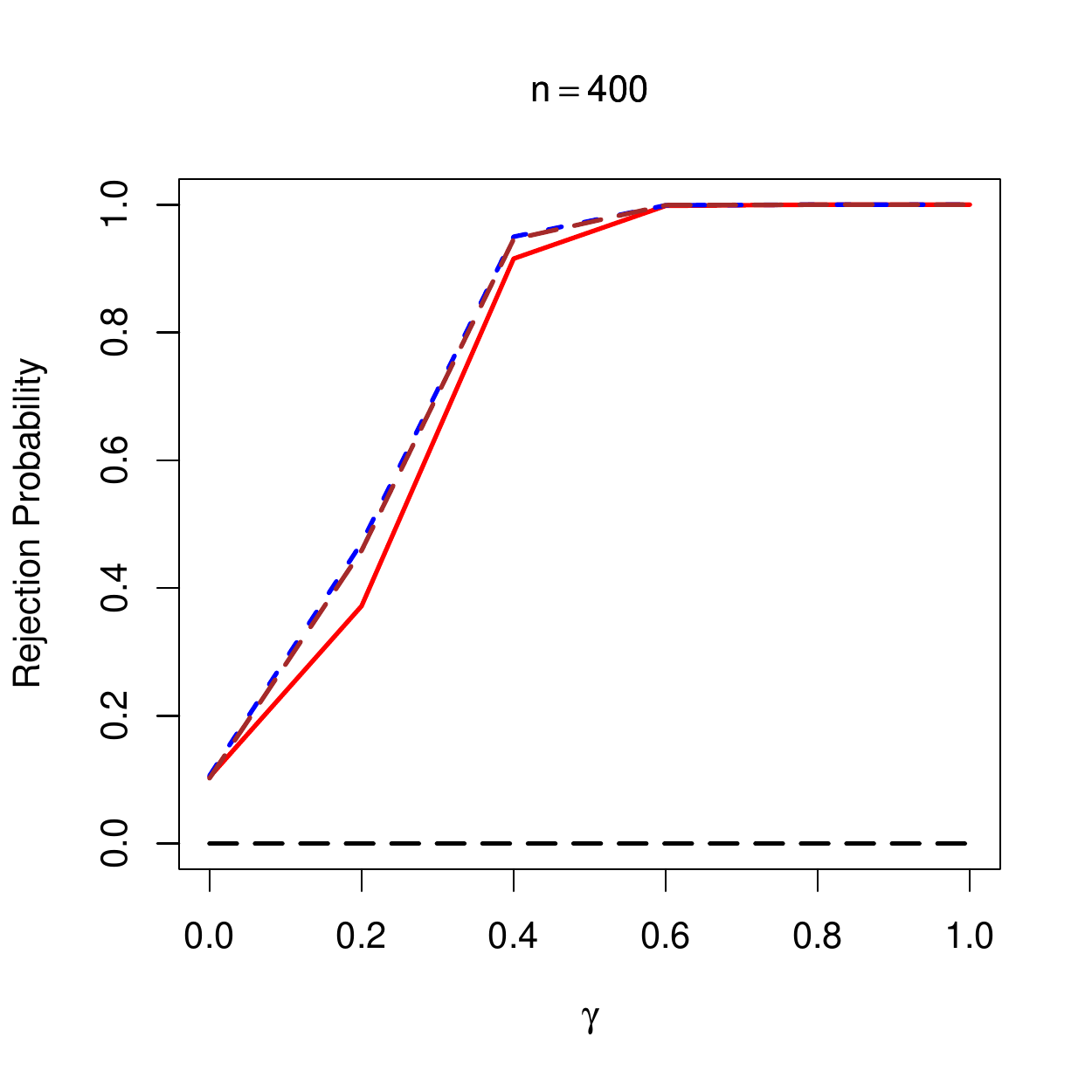}
\end{center}
\caption{\small Power curves for tests at 10\% level  using a data-driven bandwidth and no trimming.}
\label{Fig ii}
\end{figure}

\end{document}